\documentclass[a4paper, 10pt, reqno, DIV=calc]{amsart}
\usepackage[utf8]{inputenc}
\usepackage[T1]{fontenc}
\usepackage{lmodern}
\usepackage{microtype}
\usepackage{longtable}
\usepackage{amsmath, amsthm, amssymb, thmtools}
\usepackage{relsize}
\usepackage{mathrsfs}								%Script-Buchstaben
\usepackage{esint}									%Mittelwertintegralem
\usepackage{enumitem}								%Elegante Nummerierung
\usepackage{listings}								%Code-Listings im Text
\usepackage{stmaryrd}								%Parallel Transport Slash
\usepackage[mathscr]{eucal}

\usepackage[dvipsnames]{xcolor}
\usepackage[linktocpage]{hyperref}
\definecolor{colorred}{HTML}{B00000}
\definecolor{colorgreen}{HTML}{258300}
\definecolor{colorblue}{HTML}{2e32fa}
\hypersetup{colorlinks=true, linkcolor=colorred, citecolor=colorgreen, urlcolor=black}								%Farbige Links
\usepackage{xifthen}								%Für \isempty
\usepackage{url}
\usepackage{booktabs}								%Tabellen mit top-, mid- und bottomrule
\usepackage{todonotes}
\usepackage{caption}
\usepackage{bbm}									%Indikatorfunktion
\usepackage{stmaryrd}								%Doppelslash
\makeatletter
\newcommand\MyAutoefPhrasecolorGroup[1]{%
  \color@begingroup\color{MyCurrentcolor}#1\endgroup
}%
\def\HyRef@testreftype#1.#2\\{%
 \colorlet{MyCurrentcolor}{.}%
 \ltx@IfUndefined{#1autorefname}{%
   \ltx@IfUndefined{#1name}{%
     \HyRef@StripStar#1\\*\\\@nil{#1}%
     \ltx@IfUndefined{\HyRef@name autorefname}{%
       \ltx@IfUndefined{\HyRef@name name}{%
         \def\HyRef@currentHtag{}%
         \Hy@Warning{No autoref name for `#1'}%
       }{%
         \edef\HyRef@currentHtag{%
           \noexpand\MyAutoefPhrasecolorGroup{%
             \expandafter\noexpand\csname\HyRef@name name\endcsname
           }%
           \noexpand~%
         }%
       }%
     }{%
       \edef\HyRef@currentHtag{%
         \noexpand\MyAutoefPhrasecolorGroup{%
           \expandafter\noexpand
           \csname\HyRef@name autorefname\endcsname
         }%
         \noexpand~%
       }%
     }%
   }{%
     \edef\HyRef@currentHtag{%
       \noexpand\MyAutoefPhrasecolorGroup{%
         \expandafter\noexpand\csname#1name\endcsname
       }%
       \noexpand~%
     }%
   }%
 }{%
   \edef\HyRef@currentHtag{%
     \noexpand\MyAutoefPhrasecolorGroup{%
       \expandafter\noexpand\csname#1autorefname\endcsname
     }%
     \noexpand~%
   }%
 }%
}%
\makeatother

\numberwithin{equation}{section}
\newcommand{\nlc}{{\ensuremath{\textnormal{c}}}}
\newcommand{\nld}{{\ensuremath{\textnormal{d}}}}
\newcommand{\nlC}{{\ensuremath{\textnormal{C}}}}

\newcommand{\rmc}{{\ensuremath{\mathrm{c}}}}
\newcommand{\rmd}{{\ensuremath{\mathrm{d}}}}
\newcommand{\rme}{{\ensuremath{\mathrm{e}}}}

\newcommand{\rmo}{{\ensuremath{\mathrm{o}}}}

\newcommand{\rmE}{{\ensuremath{\mathrm{E}}}}

\newcommand{\sfd}{{\ensuremath{\mathsf{d}}}}
\newcommand{\sfe}{{\ensuremath{\mathsf{e}}}}

\newcommand{\sfr}{{\ensuremath{\mathsf{r}}}}

\newcommand{\scrD}{{\ensuremath{\mathscr{D}}}}
\newcommand{\scrE}{{\ensuremath{\mathscr{E}}}}
\newcommand{\scrF}{{\ensuremath{\mathscr{F}}}}

\newcommand{\scrK}{{\ensuremath{\mathscr{K}}}}

\newcommand{\scrP}{{\ensuremath{\mathscr{P}}}}
\newcommand{\scrQ}{{\ensuremath{\mathscr{Q}}}}

\newcommand{\scrU}{{\ensuremath{\mathscr{U}}}}
\newcommand{\scrV}{{\ensuremath{\mathscr{V}}}}

\newcommand{\bdalpha}{{\ensuremath{\boldsymbol{\alpha}}}}
\newcommand{\bdbeta}{{\ensuremath{\boldsymbol{\beta}}}}

\newcommand{\bdpi}{{\ensuremath{\boldsymbol{\pi}}}}

\newcommand{\bdsigma}{{\ensuremath{\boldsymbol{\sigma}}}}

\newcommand{\N}{\boldsymbol{\mathrm{N}}}						%Natürliche Zahlen
\newcommand{\Q}{\boldsymbol{\mathrm{Q}}}						%Rationale Zahlen
\newcommand{\R}{\boldsymbol{\mathrm{R}}}						%Reelle Zahlen
\renewcommand{\d}{\,\mathrm{d}}				%Kleines Integral-d

\let\limsup\undefined
\let\liminf\undefined

\DeclareMathOperator*{\limsup}{limsup}		%Limes superior
\DeclareMathOperator*{\liminf}{liminf}		%Limes inferior
\DeclareMathOperator*{\esssup}{esssup}		%Essentielles Supremum
\DeclareMathOperator*{\argmax}{argmax}		%Maximales Argument
\DeclareMathOperator{\supp}{spt}			%Träger
\let\originalleft\left			
\let\originalright\right
\renewcommand{\left}{\mathopen{}\mathclose\bgroup\originalleft}
\renewcommand{\right}{\aftergroup\egroup\originalright}

\newcommand{\mapdef}[3][]{\ifthenelse{\isempty{#1}}{#2\quad\longmapsto\quad #3}{#1\colon\quad #2\quad\longmapsto\quad #3}}		%Abbildungsdefinition, abgesetzt
\newcommand{\der}[2][]{\ifthenelse{\isempty{#1}}{\frac{\nld}{\nld #2}}{\left.\frac{\nld}{\nld #2}\right\vert_{#1}}}				%Ableitungsoperator
\makeatletter
\newcommand{\checknarg}{\@ifnextchar\bgroup{\gobblenarg}{}}
\newcommand{\gobblenarg}[1]{\@ifnextchar\bgroup{,\ \! #1\gobblenarg}{,\ \! #1}}

\theoremstyle{definition}
\newtheorem{bump}{Bump}[section]

\theoremstyle{plain}
\newtheorem{theorem}[bump]{Theorem}
\newtheorem{proposition}[bump]{Proposition}
\newtheorem{definition}[bump]{Definition}
\newtheorem{lemma}[bump]{Lemma}
\newtheorem{corollary}[bump]{Corollary}
\newtheorem{assumption}[bump]{Assumption}

\theoremstyle{remark}
\newtheorem{remark}[bump]{Remark}
\newtheorem{example}[bump]{Example}

\newtheoremstyle{cited}
{\topsep}		%Space above
{\topsep}		%Space below
{\itshape}		%Body font
{}				%Indent amount
{\bfseries}		%Theorem head font
{\textbf{.}}	%Punctuation after Theorem head
{.5em}			%Space after Theorem head
{\thmname{#1} \thmnumber{#2} \thmnote{\normalfont#3}}		%Theorem head

\theoremstyle{cited}			%Umgebungen für Zitate (d.h. ohne runde Klammern)

\makeatletter
\let\@fnsymbol\@arabic	 		%Arabische Ziffern auf der Titelseite
\makeatother

\makeatletter
\def\nonumberfootnote{\xdef\@thefnmark{}\@footnotetext}			%Fußnote ohne Nummer
\makeatother

\newcommand{\mms}{\mathit{M}}				%Basisraum
\newcommand{\met}{\sfd}						%Metrik
\newcommand{\meas}{\mathfrak{m}}			%Referenzmaß
\newcommand{\OptTGeo}{\mathrm{OptTGeo}}
\newcommand{\ac}{{\mathrm{ac}}}

\newcommand{\RCD}{\mathrm{RCD}}				%RCD-Bedingung
\newcommand{\CD}{\mathrm{CD}}				%CD-Bedingung
\newcommand{\TCD}{\mathrm{TCD}}
\newcommand{\wTCD}{\mathrm{wTCD}}
\newcommand{\TMCP}{\mathrm{TMCP}}
\newcommand{\comp}{\nlc}					%Kompakt
\newcommand{\Cont}{\nlC}					%Stetige Funktionen
\newcommand{\Ell}{\mathit{L}}				%Lebesgue-integrierbare Funktionen
\newcommand{\Lip}{\mathrm{Lip}}				%Lipschitz-stetige Funktionen/Kurven
\newcommand{\Geo}{\mathrm{Geo}}				%Geodätische
\newcommand{\Ch}{\scrE}						%Cheeger-Energie
\newcommand{\Dom}{\scrD}					%Domäne

\DeclareMathOperator{\Ent}{Ent}				%Entropie
\newcommand{\eval}{\sfe}					%Evaluationsabbildung
\newcommand{\Restr}{\mathrm{restr}}			%Einschränkungsabbildung
\newcommand{\push}{\sharp}					%Bildmaßoperator

\newcommand{\Int}{\mathrm{Int}}				%Integrationsabbildung
\newcommand{\Len}{\mathrm{Len}}

\newcommand{\TGeo}{\mathrm{TGeo}}

\newcommand{\thr}{\mathrm{thr}}
\newcommand{\Min}{\mathrm{Min}}
\newcommand{\TD}{\mathrm{TD}}
\newcommand{\STD}{\mathrm{STD}}

\newcommand{\mres}{\mathbin{\vrule height 1.6ex depth 0pt width
0.13ex\vrule height 0.13ex depth 0pt width 1.3ex}}

\usepackage{blindtext}

\allowdisplaybreaks

\providecommand{\bysame}{\leavevmode\hbox to3em{\hrulefill}\thinspace}

\let\oldtocsection=\tocsection

\let\oldtocsubsection=\tocsubsection

\let\oldtocsubsubsection=\tocsubsubsection

\renewcommand{\tocsection}[2]{\hspace{0em}\oldtocsection{#1}{#2}}
\renewcommand{\tocsubsection}[2]{\hspace{1em}\oldtocsubsection{#1}{#2}}
\renewcommand{\tocsubsubsection}[2]{\hspace{2em}\oldtocsubsubsection{#1}{#2}}

\newcommand{\nocontentsline}[3]{}
\newcommand{\tocless}[2]{\bgroup\let\addcontentsline=\nocontentsline#1{#2}\egroup}

\begin{document}

\title[Good geodesics satisfying the TCD condition]{Good geodesics satisfying the timelike curvature-dimension condition}%required
\author{Mathias Braun}%required
%\contrib[...]{...}%optional
\address{Fields Institute for Research in Mathematical Sciences, 
222 College Street, Toronto, Ontario M5T 3J1, Canada}%required
%\curraddr{...}%optional
\email{braun@math.toronto.edu}%optional
%\urladdr{...}%optional
%\dedicatory{...}%optional
\date{\today}%required
\thanks{The author is sincerely grateful to Robert McCann for help- and fruitful discussions and for pointing out the condition of regularity, and to Fabio Caval\-letti and Andrea Mondino for their help with the proof of \autoref{Le:Strong}. Most of this work has been completed during the author's employment at the Department of Mathematics at the University of Toronto, Canada. This research is supported in part by Canada Research Chairs Program funds and a
Natural Sciences and Engineering Research Council of Canada Grant (2020-04162) held by Robert McCann. The author  acknowledges funding by the Fields Institute for Research in Mathematical Sciences.}%optional
%\translator{...}%optional
%required
%optional

\begin{abstract} Let $(M,\mathsf{d},\mathfrak{m},\ll,\leq,\tau)$ be a causally closed, $\scrK$-globally hyperbolic, regular  measured Lorentzian geodesic space satisfying the weak timelike cur\-vature-dimension condition $\smash{\mathrm{wTCD}_p^e(K,N)}$ in the sense of Cavalletti and Mon\-dino. We prove the existence of geodesics of probability measures on $M$  which satisfy the entropic semiconvexity inequality defining $\smash{\mathrm{wTCD}_p^e(K,N)}$ and whose densities with respect to $\mathfrak{m}$ are additionally uniformly  $L^\infty$ in time. This holds apart from any nonbranching assumption. We also discuss similar results under the timelike measure-contraction property.
\end{abstract}

\maketitle
\thispagestyle{empty}

\tableofcontents

\addtocontents{toc}{\protect\setcounter{tocdepth}{1}}

\section{Introduction}

\subsection*{Lott--Sturm--Villani theory} Almost two decades ago, descriptions of synthetic Ricci curvature bounds by $K\in\R$ for metric mea\-sure spaces $(\mms,\met,\meas)$ via optimal transport were set up by Sturm \cite{sturm2006a,sturm2006b}, and independently Lott and Villani \cite{lott2009}. This leads to the so-called $\CD(K,\infty)$ spaces. Combined with the works of Ambrosio, Gigli, and Savaré \cite{ambrosio2014a,ambrosio2014b, ambrosio2015}, these have become a research area which is highly active  today. 
The strongest results in this framework have been obtained for  $\CD(K,N)$ spaces, i.e.~$\CD(K,\infty)$ spaces admitting a synthetic upper bound $N\in[1,\infty)$ on their ``dimension'' \cite{amb, bacher2010, cavalletti2021, erbar2015, gigli2015,  lott2009, sturm2006b}. The literature is too large to be cited exhaustively; let us only mention the splitting theorem \cite{gigli2013} for $\RCD(0,N)$ spaces, i.e.~infinitesimally Hilbertian $\CD(0,N)$ spaces, which in turn has lead to a good  structure theory for $\RCD(K,N)$ spaces  \cite{brue2020,giglimondino2015, mondino2019}.

In \cite{rajala2012a,rajala2012b} Rajala has proven the existence of $W_2$-geodesics of measures whose densities are uniformly $\Ell^\infty$ in time under $\CD(K,\infty)$ and $\CD(K,N)$, respectively, once their endpoints have bounded  support and $\Ell^\infty$-densities with respect to $\meas$. In addition, by \cite{ambrosiomondino2015,rajala2012b}  such a geodesic can be constructed to satisfy the defining entropic semiconvexity inequality. This generalizes  results from the smooth case \cite{cordero2001, mccann2001, vonrenesse2005}, and remarkably does not rely on any nonbranching assumption. 

Rajala's work has been an important ingredient to establish cor\-nerstone results for the theory of $\CD$ spaces. We only quote a few: a metric Brenier theorem \cite{ambrosio2014a},  existence of ``many'' test plans which in turn are used to axiomatize Sobolev calculus \cite{ambrosio2014a, gigli2015}, the Sobolev-to-Lipschitz property (without using heat flow \cite{ambrosio2014b}) which links this Eulerian calculus to the Lagrangian side \cite{gigli2013}, the proof of constancy of the dimension for $\CD(K,N)$ spaces \cite{brue2020}, the weak Poincaré inequality \cite{rajala2012b}, stability of super-Ricci flows \cite{sturm2018}, etc.

\subsection*{Nonsmooth Lorentzian geometry} Recently, ideas inspired in particular by \cite{erbar2015} have lead to synthetic \emph{timelike} lower Ricci bounds in nonsmooth general relativity by Cavalletti and Mondino \cite{cavalletti2020} in terms of the  \emph{\textnormal{(}weak\textnormal{)} timelike curvature-dimension conditions} $\smash{\TCD_p^e(K,N)}$ and $\smash{\wTCD_p^e(K,N)}$, $p\in (0,1)$, in the ``entropic'' sense. 

Here, the relevant structures are  \emph{Lorentzian pre-length spaces} $(\mms,\met,\ll,\leq,\tau)$  \cite{kunzinger2018}, see also \cite{sormani2016} for a related approach. In the same way metric measure spaces generalize smooth Riemannian manifolds, these provide singular  counterparts to smooth Lorentzian spacetimes. They come with a chronological relation $\ll$ and a notion of causality $\leq$ between points in $\mms$. The \emph{time separation function} $\tau$ takes over the role of a metric, in the sense that parallel to the smooth case \cite{hawking1973, mccann2020, oneill1983} it allows for notions of length, geodesics, strong causality, etc.\footnote{Though $\tau(x,y)$ should not be interpreted as a distance, cf.~\autoref{Re:Ro}, but rather as the maximal proper time a spacetime  point $x\in\mms$ needs to travel to $y\in\mms$.} Developing such a singular theory within general relativity aims to include spacetimes with low regularity metrics. In turn, this would allow one to address the Cauchy initial value problem for the Einstein equation, the cosmic censorship conjectures, and physically relevant models in wider generality than currently possible. See e.g.~the introductions of \cite{cavalletti2020,kunzinger2018} for related literature. %This is desirable since in fact, singularities are physically  postulated to appear, e.g.~in black hole interiors.

The $\TCD$ and $\wTCD$ conditions are defined by asking for the  convexity, see  \autoref{Def:Convex} and \autoref{Def:TCD}, of the exponentiated relative entropy
\begin{align}\label{Eq:UN def}
\scrU_N(\mu) := \rme^{-\Ent_\meas(\mu)/N},
\end{align}
on the space of probability measures on $\mms$, along some  chronological optimal mass transport from past  to future located distributions. Here, optimality --- and the notion of \emph{geodesics} with respect to which convexity of $\scrU_N$ is  formulated --- are quantified by the \emph{$p$-Lorentz--Wasserstein distance} 
\begin{align*}
\ell_p(\mu,\nu) := \sup  \Vert\tau\Vert_{\Ell^p(\mms^2,\pi)},
\end{align*}
$p\in(0,1]$, first introduced in \cite{eckstein2017} and further studied in \cite{cavalletti2020,kell2020,mccann2020,mondinosuhr2018,suhr2016}. The supremum is taken over all causal couplings $\pi$ of $\mu$ and $\nu$; see \autoref{Sub:Lorentz}  for details. 

Notably, $\smash{\TCD_p^e(0,N)}$ and $\smash{\wTCD_p^e(0,N)}$ are nonsmooth analogues of the  \emph{strong energy condition} of Hawking and Penrose \cite{hawking1966,hawking1970,penrose1965}. The latter is a nonnegative definiteness condition on the stress-energy tensor. For a vanishing cosmological constant, this boils down to nonnegativity of the Ricci tensor in every timelike direction by the Einstein equation. In turn, the latter property can be characterized by convexity of $\scrU_N$ along suitable $\smash{\ell_p}$-geodesics. This was discovered in \cite{mccann2020,mondinosuhr2018} and in fact motivated the authors of \cite{cavalletti2020} to introduce the $\TCD$ and $\wTCD$ conditions for singular \emph{measured Lorentzian pre-length spaces} $(\mms,\met,\meas,\ll,\leq,\tau)$, i.e.~Lorentzian pre-length spaces endowed with a reference measure $\meas$. 

\subsection*{Objective} In this article, we provide a nonsmooth  Lorentzian version of Ra\-jala's re\-sults \cite{rajala2012a,rajala2012b}, namely the rich existence of \emph{good} $\smash{\ell_p}$-geodesics. We hope our results to be an equally useful contribution to the young theory of $\TCD$ and $\wTCD$ spaces as \cite{rajala2012a,rajala2012b} was for $\CD$ spaces; possible applications of \autoref{Th:Linfty estimates} and related future work we attack soon are described  below.

Before stating our major  \autoref{Th:Linfty estimates}, we specify what we mean by ``good''. A detailed account of our  notation is postponed to \autoref{Ch:Opt}. 

All over this paper, given $\pi\in\scrP(\mms^2)$ we  use the abbreviation
\begin{align*}
T_\pi := \Vert\tau\Vert_{\Ell^2(\mms,\pi)}.
\end{align*}

\begin{definition}\label{Def:Good} Let $(\mms,\met,\meas,\ll,\leq,\tau)$ be a measured Lorentzian pre-length space, and let $p\in (0,1)$, $K\in\R$, and $N\in (0,\infty)$. A timelike proper-time parametrized $\smash{\ell_p}$-geodesic $(\mu_t)_{t\in[0,1]}$, in a sense made precise in  \autoref{Sub:Geodesics}, is called \emph{good} if the following conditions hold.
\begin{enumerate}[label=\textnormal{\alph*.}]
\item There exists some $\smash{\ell_p}$-optimal coupling $\smash{\pi\in \Pi_\ll(\mu_0,\mu_1)}$ with $\tau\in\Ell^2(\mms^2,\pi)$, and for every $t\in [0,1]$, $\mu_t=\rho_t\,\meas\in\Dom(\Ent_\meas)$ and
\begin{align}\label{Eq:Ineq UN}
\scrU_N(\mu_t) \geq \sigma_{K,N}^{(1-t)}(T_\pi)\,\scrU_N(\mu_0) + \sigma_{K,N}^{(t)}(T_\pi)\,\scrU_N(\mu_0).
\end{align}
\item We have the uniform $\Ell^\infty$-bound
\begin{align*}
\sup\!\big\lbrace\Vert\rho_t\Vert_{\Ell^\infty(\mms,\meas)} : t\in [0,1]\big\rbrace < \infty.
\end{align*}
\end{enumerate}
\end{definition}

\begin{theorem}\label{Th:Linfty estimates} Assume that $(\mms,\met,\meas,\ll,\leq,\tau)$ is a causally closed, $\scrK$-globally hyperbolic, regular Lorentzian geodesic space satisfying $\smash{\wTCD_p^e(K,N)}$ for $p\in (0,1)$, $K\in\R$, and $N\in (0,\infty)$. Let  $(\mu_0,\mu_1)=(\rho_0\,\meas, \rho_1\,\meas)\in\scrP_\comp^\ac(\mms,\meas)^2$ be strongly timelike $p$-dua\-li\-zable, and suppose that the density $\rho_0,\rho_1\in\Ell^\infty(\mms,\meas)$. Then there exists a good timelike proper-time parametrized $\smash{\ell_p}$-geodesic from $\mu_0$ to $\mu_1$.

More precisely, there exists a timelike proper-time parametrized $\ell_p$-geodesic $(\mu_t)_{t\in [0,1]}$ from $\mu_0$ to $\mu_1$ satisfies, for every $t\in[0,1]$, $\mu_t=\rho_t\,\meas\in\Dom(\Ent_\meas)$ and
\begin{align*}
\Vert \rho_t\Vert_{\Ell^\infty(\mms,\meas)} \leq \rme^{D\sqrt{K^-N}/2}\,\max\!\big\lbrace\Vert\rho_0\Vert_{\Ell^\infty(\mms,\meas)}, \Vert\rho_1\Vert_{\Ell^\infty(\mms,\meas)} \big\rbrace,
\end{align*}
where  $D:= \sup\tau(\supp\mu_0\times\supp\mu_1)$.
\end{theorem}

We also establish the subsequent variations of \autoref{Th:Linfty estimates}.
\begin{itemize}
\item Assuming the stronger $\smash{\TCD_p^e(K,N)}$ condition in place of $\smash{\wTCD_p^e(K,N)}$, one may relax the hypotheses on $\mu_0$ and $\mu_1$, cf.~\autoref{Re:Minor}.
\item A version of it holds for a  dimension-independent $\smash{\wTCD_p(K,\infty)}$ condition, newly introduced in \autoref{Def:CD K infty} below following \cite{sturm2006a}, cf.~\autoref{Th:Ent linfty} and \autoref{Re:BL}. A key message here is that the upper dimension bound, apart from which timelike Ricci bounds have not been studied thus far, is not strictly required, but of course gives quantitatively better results.
\item Our proof can  be adapted to the situation of more general timelike convex functionals on $\scrP(\mms)$, cf.~\autoref{Th:CONV}.
\item Lastly, we present a version for the \emph{timelike measure-contraction property} $\TMCP_p^e(K,N)$ from \cite{cavalletti2020}, following the work \cite{cavalletti2017} for metric measure spaces, cf.~\autoref{Th:TMCP}. Here, by nature of the $\TMCP$ condition (the terminal measure is a Dirac measure) the r.h.s.~of the obtained inequality
\begin{align*}
\Vert\rho_t\Vert_{\Ell^\infty(\mms,\meas)} \leq \frac{1}{(1-t)^N}\,\rme^{Dt\sqrt{K^-N}}\,\Vert\rho_0\Vert_{\Ell^\infty(\mms,\meas)}
\end{align*}
is not bounded in $t\in[0,1]$, but blows up as $t$ approaches $1$.
\end{itemize}

In the $\wTCD$ context, \autoref{Th:Linfty estimates} seems optimal with respect to which endpoints $\mu_0$ and $\mu_1$ are allowed, i.e.~strongly timelike $p$-dualizable ones according to \autoref{Def:TL DUAL}. Having  covered this framework, \autoref{Th:Linfty estimates} is expected to be relevant for stability questions, since the limit of a sequence of $\TCD$ spaces  is only known to be $\wTCD$ in general  \cite[Thm.~3.14]{cavalletti2020}. We point out that local causal closedness and $\scrK$-global hyperbolicity, two main  assumptions of \autoref{Th:Linfty estimates}, are precisely the regularity conditions used to set up the corresponding notion of weak convergence of measured Lorentzian geodesic spaces in \cite{cavalletti2020}. %On the other hand, the assumption of regularity is rather technical, but seems essential for a crucial lemma in this work, namely \autoref{Le:Plans}, cf.~\autoref{Re:Regularity cond}. 
Moreover, the author is indebted to Robert McCann for pointing out to add the assumption of regularity, cf.~\autoref{REG}, to \autoref{Th:Ent linfty} and to \autoref{Cor:Timelike branching} below.

As in \cite{rajala2012a,rajala2012b} we avoid any assumption on (timelike) nonbranching \cite[Def.~1.10]{cavalletti2020}, which makes \autoref{Th:Linfty estimates} convenient for spacetimes with low regularity. Indeed, while spacetimes with $\smash{\Cont^{1,1}}$-metrics are timelike nonbranching \cite{cavalletti2020}, this property is expected to fail e.g.~for lower regularity of the metric and for closed cone structures \cite{minguzzi2019}. Nevertheless, already in the timelike nonbranching case, an interesting byproduct of \autoref{Th:Linfty estimates} and the uniqueness of chronological $\smash{\ell_p}$-optimal couplings and $\smash{\ell_p}$-geodesics \cite[Thm.~3.19, Thm.~3.20]{cavalletti2020} is the following.
  
\begin{corollary}\label{Cor:Timelike branching} Assume that $(\mms,\met,\meas,\ll,\leq,\tau)$ is a timelike nonbranching,  causally closed, $\scrK$-globally hyperbolic, regular Lorentzian geodesic space obeying $\smash{\wTCD_p^e(K,N)}$ for $p\in (0,1)$, $K\in\R$, and $N\in (0,\infty)$. Let the pair $(\mu_0,\mu_1) = (\rho_0\,\meas,\rho_1\,\meas)\in\scrP_\comp^\ac(\mms,\meas)^2$  be strongly timelike $p$-dualizable, and suppose that $\rho_0,\rho_1\in\Ell^\infty(\mms,\meas)$. Then the unique timelike proper-time parametrized $\smash{\ell_p}$-geo\-desic $(\mu_t)_{t\in[0,1]}$ from $\mu_0$ to $\mu_1$ is good.
\end{corollary}

\begin{remark} One can replace timelike nonbranching by the weaker condition of timelike $p$-essential nonbranching \cite[Def.~2.21, Rem.~2.22]{braun+} in \autoref{Cor:Timelike branching}. In fact, in this case, independently of the results in this paper we have recently shown the equivalence of the conditions $\smash{\TCD_p^e(K,N)}$ and $\smash{\wTCD_p^e(K,N)}$ \cite[Thm.~3.35]{braun+} (and to the reduced $\TCD$ conditions in terms of \emph{Rényi's entropy} from \cite[Def.~3.2]{braun+}). The above mentioned uniqueness results still hold \cite[Thm.~4.15, Thm.~4.16, Rem.~4.17]{braun+}.
\end{remark}

\subsection*{Proof strategy for \autoref{Th:Linfty estimates}} The proof of  \autoref{Th:Linfty estimates} is fully carried out in \autoref{Ch:Good}. We were mostly inspired by the clever strategy \cite{rajala2012b} whose arguments we adapt  to the Lorentzian setting at various instances.

Roughly speaking, the proof is based on a bisection argument. Given $\mu_0$ and $\mu_1$ as in the assumptions, we choose an $\smash{\ell_p}$-midpoint $\mu_{1/2}$ of them at which $\scrU_N$ is maximal, see \autoref{Le:Maximizer}. By   $\smash{\wTCD_p^e(K,N)}$, it obeys \eqref{Eq:Ineq UN} for time $1/2$. Then we choose $\smash{\ell_p}$-mid\-points $\mu_{1/4}$ of $\mu_0$ and $\mu_{1/2}$, and $\mu_{3/4}$ of $\mu_{1/2}$ and $\mu_1$, in the same manner. Iteratively, we thus construct a collection of measures $\mu_t\in\Dom(\Ent_\meas)$ for every dyadic $t\in[0,1]$. Based on a crucial property of the distortion coefficients $\smash{\sigma_{K,N}^{(r)}}$, stated in \autoref{Le:Distortion}, the inequality \eqref{Eq:Ineq UN} --- which  is a priori only true for $\mu_t$ at time $1/2$ with $\mu_0$ and $\mu_1$ replaced by its ancestors --- extends to every dyadic $t\in[0,1]$ with $\mu_0$ and $\mu_1$ on the r.h.s. This inequality is stable under weak completion to a timelike proper-time parametrized $\smash{\ell_p}$-geodesic $(\mu_t)_{t\in[0,1]}$ defined on all of $[0,1]$, obtained by intermediate gluing of timelike $\smash{\ell_p}$-optimal geodesic plans, and passage to the limit.

As afterwards discussed in \autoref{Sub:Uniform density bounds}, maximality of $\scrU_N$ already suffices to guarantee the desired density bounds for $(\mu_t)_{t\in[0,1]}$. The   reason is that by the curvature-dimension condition, for appropriate endpoints mass has to be spread along some timelike proper-time parametrized $\smash{\ell_p}$-geodesic in a certain way,  see \autoref{Le:Bounded midpoints}. If the density of some midpoint in our  construction was too large, we could use the latter  geodesic to shuffle mass around and to build a midpoint with strictly smaller entropy, contradicting the maximality of $\scrU_N$ at our chosen midpoint. The proofs of the key \autoref{Pr:Min 0} and \autoref{Pr:Max = 0} are based on this principle.

Contrary to \cite{rajala2012a,rajala2012b}, however, we do not really work with bare $\smash{\ell_p}$-intermediate points, but  with intermediate slices of timelike $\smash{\ell_p}$-optimal geodesic \emph{plans}, cf.~\autoref{Sub:Geodesics}. An $\smash{\ell_p}$-intermediate point does not need to admit a \emph{chronological} coupling to any of its endpoints, which is however required to invoke the $\wTCD$ condition. On the other hand, this is clear if the chosen intermediate point already lies on a timelike $\smash{\ell_p}$-optimal geodesic plan. In fact, the  notion of \emph{strong} timelike $p$-dualizability is preserved along such plans, see \autoref{Le:Strong}. (For timelike $p$-dualizability arising in the definition of $\TCD_p^e(K,N)$, however, this is unclear, cf.~\autoref{Re:Minor}.)

\subsection*{Organization} In \autoref{Ch:Opt}, we shortly review the theories of Lorentzian pre-length spaces, Lorentzian optimal transport, and timelike curvature-dimension conditions. \autoref{Ch:Good} contains the proof of \autoref{Th:Linfty estimates}, while \autoref{Ch:Variations}  outlines, and partially proves, the above mentioned extensions of \autoref{Th:Linfty estimates}.

\addtocontents{toc}{\protect\setcounter{tocdepth}{2}}

\section{Optimal transport on Lorentzian spaces}\label{Ch:Opt}

This chapter recalls recent progress in nonsmooth Lorentzian geometry and optimal transport theory on such spaces. The reader is invited to consult the main references  \cite{cavalletti2020, kunzinger2018} for more details, proofs, and especially examples.

\subsection{Lorentzian geodesic spaces}\label{Sec:Lorentzian geo}

\subsubsection{Basic assumptions and notation}\label{Sub:Basic} Everywhere in this paper, let $(\mms,\met)$ be a proper ---  hence complete and separable --- metric space. All topological and measure-theoretic properties are understood with respect to the topology induced by $\met$ and its induced Borel $\sigma$-algebra, respectively. 

Moreover, we always fix a nontrivial Radon measure $\meas$ on $\mms$.  For simplicity, we assume  that $\meas$ is fully supported, in symbols $\supp\meas =\mms$. 

Let $\scrP(\mms)$ be the set of all probability measures on $\mms$. Let $\scrP_\comp(\mms)$ and $\scrP^\ac(\mms,\meas)$ be its subsets of all elements with compact support and $\meas$-absolutely continuous measures, respectively, and set $\smash{\scrP_\comp^\ac(\mms,\meas) := \scrP_\comp(\mms)\cap\scrP^\ac(\mms,\meas)}$. Whenever we say that a specified property is satisfied ``subject to the decomposition $\mu = \rho\,\meas + \mu_\perp$'', we mean that $\mu_\perp$ is the $\meas$-singular part in the Lebesgue decomposition of $\mu\in\scrP(\mms)$ with respect to $\meas$, and that $\mu - \mu_\perp = \rho\,\meas\in\scrP^\ac(\mms,\meas)$.

For a Borel map $F\colon \mms\to \mms'$ into a metric space $(\mms',\met')$ as well as $\mu\in\scrP(\mms)$, $F_\push\mu \in\scrP(\mms')$ designates the usual \emph{push-forward}  of $\mu$ under $F$ given by $F_\push\mu[B] := \smash{\mu\big[F^{-1}(B)\big]}$ for every Borel set $B\subset\mms'$.

Given  $\mu,\nu\in\scrP(\mms)$, let $\Pi(\mu,\nu)$ be the set of all \emph{couplings} of $\mu$ and $\nu$, i.e.~all $\pi\in\scrP(\mms^2)$ with $\pi[\, \cdot\times\mms] = \mu$ and $\pi[\mms\times\cdot\, ] = \nu$.

Let $\Cont([0,1];\mms)$ denote the set of all continuous curves $\gamma\colon [0,1]\to\mms$, endowed with the uniform topology. For $t\in [0,1]$, the \emph{evaluation map} $\eval_t\colon \Cont([0,1];\mms) \to \mms$ is defined by $\eval_t(\gamma) := \gamma_t$.

\subsubsection{Chronology and causality} Throughout, we fix a preorder $\leq$  and a transitive relation $\ll$, contained in $\leq$, on $\mms$. We say that $x,y\in\mms$ are in \emph{timelike} or \emph{causal} relation if $x\ll y$ or $x\leq y$, respectively. The triple $(\mms,\ll,\leq)$ is called \emph{causal space} \cite[Def.~2.1]{kunzinger2018}. We write $x<y$ provided $x\leq y$ yet $x\neq y$. Let us set
\begin{align*}
\mms_\ll^2 &:= \{(x,y) \in\mms^2 : x\ll y\},\\
\mms_\leq^2 &:= \{(x,y)\in\mms^2 : x\leq y\}.
\end{align*}

%The following notion \cite[Def.~3.4]{kunzinger2018} is a key concept for us.

\begin{definition} We term a causal space $(\mms,\ll,\leq)$  \emph{causally closed} if $\leq$ is closed, i.e.~$\smash{\mms_\leq^2}$ is closed in $\mms^2$. 
\end{definition}

Given a (not necessarily Borel) set $A\subset\mms$, we define \cite[Def.~2.3]{kunzinger2018}  the \emph{chro\-no\-logical future} $I^+(A)\subset\mms$ and the \emph{causal future} $J^+(A)\subset\mms$ of $A$ by
\begin{align*}
I^+(A) &:= \{y\in\mms : x\ll y\textnormal{ for some }x\in A\},\\
J^+(A) &:= \{y\in\mms : x\leq y\textnormal{ for some }x\in A\}.
\end{align*}
Analogously, the \emph{chronological past} $I^-(A)$ and the \emph{causal past} $J^-(A)$ of $A$ are defined. By a slight abuse of notation, given  $x\in\mms$ we write $\smash{I^\pm(x) := I^\pm(\{x\})}$ and $\smash{J^\pm(x) := J^\pm(\{x\})}$. For $\mu\in\scrP(\mms)$, we write $\smash{I^\pm(\mu) := I^\pm(\supp\mu)}$ and $\smash{J^\pm(\mu)} := \smash{J^\pm(\supp\mu)}$. For all these objects, we set $I(A,B) := I^+(A) \cap I^-(B)$, and we define $I(x,y)$, $I(\mu,\nu)$, $J(A,B)$, $J(x,y)$, and $J(\mu,\nu)$ analogously. 

\subsubsection{Lorentzian pre-length spaces}  A function $\tau\colon\mms^2 \to [0,\infty]$ is a \emph{time separation function} \cite[Def.~2.8]{kunzinger2018} if it is lower semicontinuous, and for every $x,y,z\in\mms$,
\begin{enumerate}[label=\textnormal{\alph*.}]
\item $\tau(x,y) = 0$ if $x\not\leq y$,
\item $\tau(x,y)>0$ if and only if $x\ll y$  --- in other words, $\smash{\mms_\ll^2 = \{\tau > 0\}}$  --- and
\item if $x\leq y\leq z$ we have the \emph{reverse triangle inequality}
\begin{align}\label{Eq:Reverse tau}
\tau(x,z) \geq \tau(x,y) + \tau(y,z).
\end{align}
\end{enumerate}
The existence of such a $\tau$ implies that $\ll$ is an \emph{open} relation \cite[Prop.~2.13]{kunzinger2018}, whence $\smash{I^\pm(A)}$ is open for every $A\subset\mms$ \cite[Lem.~2.12]{kunzinger2018}.

\begin{remark}\label{Re:Ro} Besides \eqref{Eq:Reverse tau}, unlike the metric in metric geometry $\tau$ is only symmetric in pathological cases: for every $x\in\mms$, either $\tau(x,x)=0$ or $\tau(x,x) = \infty$, and if $\tau(x,y)\in (0,\infty)$, then $\tau(y,x) = 0$ \cite[Prop.~2.14]{kunzinger2018}.
\end{remark}

\begin{definition} A \emph{Lorentzian pre-length space} is a quintuple $(\mms,\met,\ll,\leq,\tau)$ consisting of a causal space $(\mms,\ll,\leq)$ equipped with a proper metric $\met$ and a time separation function $\tau$ as above.
\end{definition}

\subsubsection{Length of curves} Let $(\mms,\met,\ll,\leq,\tau)$ be a Lorentzian pre-length space. A \emph{curve} is a continuous map $\gamma\colon [a,b]\to\mms$, $a,b\in\R$ with $a<b$. Such a curve $\gamma$ is called \emph{\textnormal{(}future-directed\textnormal{)} timelike} or \emph{\textnormal{(}future-directed\textnormal{)} causal} if it is Lipschitz continuous with respect to $\met$, and $\gamma_s \ll \gamma_t$ or $\gamma_s \leq \gamma_t$ for every $s,t\in[a,b]$ with $s< t$, respectively. %(For the purposes of our paper, it would be sufficient to assume that every curve $\gamma$ with $\gamma_s\ll\gamma_t$ for every $s,t\in[a,b]$ with $s<t$ admits a $\met$-Lipschitz reparametrization, cf.~\cite[Sec.~3.7]{kunzinger2018}.) 
It is  \emph{null} if it is causal and $\gamma_a \not\ll\gamma_b$. Analogous notions make sense for  \emph{past-directed} curves and their \emph{causal character} (i.e.~their property of being chronological, causal, or null). Unless stated otherwise, every curve of a specified causal character is assumed future-directed. %Every causal curve $\gamma\colon[a,b]\to\mms$ is Lipschitz continuous with respect to $\met$ and $\met$-rectifiable \cite[Lem.~2.19]{kunzinger2018}.

The \emph{length} of a curve $\gamma\colon [a,b]\to \mms$ is defined through
\begin{align*}
\Len_\tau(\gamma) := \inf\!\big\lbrace\tau(\gamma_{t_0},\gamma_{t_1}) + \dots + \tau(\gamma_{t_{n-1}},\gamma_{t_n})\big\rbrace, 
\end{align*}
where the infimum is taken over all $n\in\N$ and all $t_0,\dots,t_n \in[a,b]$ with $t_0=0$, $t_n=1$, and $t_i < t_{i+1}$ for every $i\in\{0,\dots,n-1\}$ \cite[Def.~2.24]{kunzinger2018}. It is additive with respect to restriction to disjoint partitions \cite[Lem.~2.25]{kunzinger2018}. Reparametrizations do neither change causal characters  \cite[Lem.~2.27]{kunzinger2018} nor the $\tau$-length \cite[Lem.~2.28]{kunzinger2018}.

\subsubsection{Geodesics}\label{Sub:GEO} A future-directed causal curve $\gamma\colon[a,b]\to\mms$ is a \emph{geodesic} (or \emph{maximal}) provided $\Len_\tau(\gamma) = \tau(\gamma_a,\gamma_b)$ \cite[Def.~2.33]{kunzinger2018}. The spaces of all such curves is denoted by $\Geo(\mms)\subset\Lip([0,1];\mms)$, and its subset of timelike curves is called $\TGeo(\mms)$. 

Recall that every element of $\TGeo(\mms)$ has a weak parametrization \cite[Def.~3.31]{kunzinger2018} by $\tau$-arclength if $\tau$ is continuous and $\tau(x,x) = 0$ for every $x\in\mms$ \cite[Cor.~3.35]{kunzinger2018}. This induces a natural reparametrization map $\sfr\colon\TGeo(\mms)\to\Cont([0,1];\mms)$ which is continuous \cite[Lem.~B.6]{braun+}. In particular, all elements $\eta$ of 
\begin{align}\label{Eq:TGeo tau}
\TGeo^\tau(\mms) := \sfr(\TGeo(\mms))
\end{align}
are timelike and proper-time parametrized, i.e.~for every $s,t\in[0,1]$ with $s<t$,
\begin{align*}
\tau(\eta_s,\eta_t) = (t-s)\,\tau(\eta_0,\eta_1)>0.
\end{align*}
In general, elements of $\smash{\TGeo^\tau(\mms)}$ are \emph{not} Lipschitz continuous \cite[p.~424]{kunzinger2018}. Yet, as $\smash{\TGeo^\tau(\mms)}$ is the continuous image of $\TGeo(\mms)$ which, as subset of $\Lip([0,1];\mms)$, will have better compactness properties under \autoref{Ass:Ass}, cf.~\autoref{Sub:GH}, the latter transfer to $\TGeo^\tau(\mms)$; compare with \cite[Sec. B.3]{braun+} and \autoref{Le:Plans}.

We call $(\mms,\met,\ll,\leq,\tau)$ \emph{geodesic} if for every $x,y\in\mms$ with $x<y$ there exists a future-directed causal geodesic $\gamma\in\Geo(\mms)$ with initial point $x$ and final point $y$.

\subsubsection{Regularity}\label{REG} We will later assume $(\mms,\met,\ll,\leq,\tau)$ to be \emph{regular}(\emph{ly localizable}) according to \cite[Def.~3.16]{kunzinger2018}. Instead of giving this rather technical definition, let us list its most important consequences which will be relevant for our purposes. Under regularity, geodesy, and global hyperbolicity, cf.~\autoref{Sub:GH} below, the $\tau$-length $\Len_\tau$ is upper semicontinuous in the following sense \cite[Prop.~3.17, Thm.~3.26]{kunzinger2018}: if $(\gamma_n)_{n\in\N}$ is a sequence of causal curves $\gamma_n \colon[0,1]\to\mms$ converging uniformly to a causal curve $\gamma\colon [0,1]\to\mms$, then $\smash{\Len_\tau(\gamma) \geq \limsup_{n\to\infty} \Len_\tau(\gamma_n)}$. Moreover, in regular Lorentzian pre-length spaces, causal geodesics have a causal character, i.e.~they are either timelike or null (hence do not switch from timelike to causal or vice versa) \cite[Thm.~3.18]{kunzinger2018}. In this case, a causal geodesic $\gamma\colon[0,1]\to\mms$ is timelike if and only if $\tau(\gamma_0,\gamma_1)>0$. In particular, on regular Lorentzian geodesic spaces, every $x,y\in \mms$ with $x\ll y$ can be connected by a \emph{timelike} geodesic --- which does \emph{not} follow from geodesy alone --- and the set of all such geodesics is closed for fixed $x$ and $y$ \cite[Lem.~B.1]{braun+}.

\subsubsection{Global hyperbolicity}\label{Sub:GH} Next, we introduce a   useful condition which ensures both that the time separation function $\tau$ behaves nicely, and that geodesics exist. 

Following \cite[Sec.~1.1]{cavalletti2020}, we call $(\mms,\met,\ll,\leq,\tau)$  \emph{non-totally imprisoning} if for every compact $C\subset \mms$ there exists a constant  $c>0$ such that the arclength of every causal curve in $C$ with respect to $\met$ is bounded from above by $c$.

\begin{definition} A Lorentzian pre-length space $(\mms,\met,\ll,\leq,\tau)$ is  \emph{globally hyperbolic} if it is non-totally imprisoning and  $J(x,y)$  is compact for every $x,y\in\mms$. It is \emph{$\scrK$-globally hyperbolic} if $J(K_0,K_1)$ is compact for all compact $K_0,K_1\subset\mms$.
\end{definition}

\begin{remark} $\scrK$-global hyperbolicity is not much more restrictive than global hyperbolicity: if in addition to the latter, $(\mms,\met,\ll,\leq,\tau)$ is locally causally closed \cite[Def.~3.4]{kunzinger2018} and $I^\pm(x) \neq \emptyset$ for every $x\in \mms$, $\scrK$-global hyperbolicity holds true  \cite[Lem.~1.5]{cavalletti2020}.

On the other hand, every locally causally closed, $\scrK$-globally hyperbolic Lorentzian geodesic space is in fact causally closed \cite[Lem.~1.6]{cavalletti2020}.
\end{remark}

Thanks to \cite[Def.~3.25, Thm.~3.26]{kunzinger2018}, global hyperbolicity implies the nonsmooth analogue of the well-known \emph{strong causality condition} for smooth Lorentzian spacetimes \cite[Def.~14.11]{oneill1983}. On every globally hyperbolic Lorentzian length space (see  \cite[Def.~3.22]{kunzinger2018} for the precise definition),  $\tau$ is finite and continuous \cite[Thm.~3.28]{kunzinger2018}; in particular, $\tau(x,x) = 0$ for every $x\in\mms$. Also, every such space is geodesic by the nonsmooth Avez--Seifert theorem \cite[Thm.~3.30]{kunzinger2018}. 

In the recent work \cite{burtscher2021}, a nonsmooth analogue of Geroch's characterization \cite{geroch1970} of global hyperbolicity in terms of Cauchy time functions has been achieved.

$\scrK$-global hyperbolicity is  convenient for optimal transport purposes later.

\subsection{Lorentz--Wasserstein distance}\label{Sub:Lorentz} Smooth Lorentzian theories of optimal transport have been studied in \cite{eckstein2017,kell2020, mccann2020, mondinosuhr2018, suhr2016}. We review the cornerstones of the accom\-panying nonsmooth theory recently developed in \cite{cavalletti2020} now.

\subsubsection{Chronological and causal couplings} Let $(\mms,\met,\ll,\leq,\tau)$ be a Lorentzian pre-length space, and $\mu,\nu\in\scrP(\mms)$. We define the set $\Pi_\ll(\mu,\nu)$ of all \emph{chronological couplings} of $\mu$ and $\nu$ to consist of all $\pi\in\Pi(\mu,\nu)$ with $\smash{\pi[\mms_\ll^2]=1}$. Similarly, the set $\Pi_\leq(\mu,\nu)$ of all \emph{causal couplings} of $\mu$ and $\nu$ is defined. If $(\mms,\met,\ll,\leq,\tau)$ is causally closed, clearly $\pi\in\Pi_\leq(\mu,\nu)$ if and only if $\pi\in\Pi(\mu,\nu)$ and $\smash{\supp\pi\subset\mms_\leq^2}$.

Intuitively, a chronological or causal coupling of $\mu$ and $\nu$ describes a law for transporting an infinitesimal mass portion $\rmd\mu(x)$ to an infinitesimal mass portion $\rmd\nu(y)$ in such a way that $x\ll y$ or $x\leq y$, respectively.

\subsubsection{The $\smash{\ell_p}$-optimal transport problem} In the following, given any $p\in (0,1)$ and following \cite{cavalletti2020,mccann2020} we adopt the  conventions
\begin{align*}
\sup\emptyset &:= -\infty,\\
(-\infty)^p &:= (-\infty)^{1/p} := -\infty,\\
\infty -\infty &:=-\infty. \textcolor{white}{a^{1/p}}
\end{align*}

The total transport cost function $\ell_p\colon\scrP(\mms)^2\to [0,\infty]\cup\{-\infty\}$  is given by
\begin{align*}
\ell_p(\mu,\nu) &:=  \sup\!\big\lbrace \Vert \tau\Vert_{\Ell^p(\mms^2,\pi)} : \pi \in \Pi_\leq(\mu,\nu)\big\rbrace\\
&\textcolor{white}{:}= \sup\!\big\lbrace \Vert l\Vert_{\Ell^p(\mms^2,\pi)} : \pi\in\Pi(\mu,\nu)\big\rbrace.
\end{align*}
Here, the function $l\colon \mms^2 \to [0,\infty]\cup\{-\infty\}$ is defined by
\begin{align*}
l^p(x,y)  := \begin{cases} \tau^p(x,y) & \textnormal{if }x\leq y,\\
-\infty & \textnormal{otherwise}.
\end{cases}
\end{align*}

\begin{remark}
The sets of  maximizers for both suprema defining $\smash{\ell_p(\mu,\nu)}$ coincide. One advantage of the second formulation  is that under (local) causal closedness and global hyperbolicity assumptions, $l^p$ is upper semicontinuous on $\mms^2$. In this case,  standard optimal transport techniques \cite{ambrosiogigli, villani2009} can be applied to study the second problem, which in turn yields results for the first  \cite[Rem.~2.2]{cavalletti2020}. Moreover, the preimages $l^{-1}([0,\infty))$ and $l^{-1}((0,\infty))$ conveniently  encode causality and chronology of points in $\mms^2$, respectively.
\end{remark}

A coupling $\pi\in\Pi(\mu,\nu)$ of $\mu,\nu\in\scrP(\mms)$ is \emph{$\ell_p$-optimal} if $\pi\in\Pi_\leq(\mu,\nu)$ and 
\begin{align*}
\ell_p(\mu,\nu) = \Vert\tau\Vert_{\Ell^p(\mms^2,\pi)} = \Vert l \Vert_{\Ell^p(\mms^2,\pi)}.
\end{align*}
If $(\mms,\met,\ll,\leq,\tau)$ is locally causally closed and globally hyperbolic, and if $\mu,\nu\in\scrP_\comp(\mms)$ with $\Pi_\leq(\mu,\nu)\neq\emptyset$, then there exists an $\smash{\ell_p}$-optimal coupling $\pi$ of $\mu$ and $\nu$, and its total transport cost $\smash{\Vert\tau\Vert_{\Ell^p(\mms^2,\pi)}}$ is finite \cite[Prop.~2.3]{cavalletti2020}.\footnote{This holds in more generality, but the named case will be the only relevant in our work.}

Lastly, an important property of $\smash{\ell_p}$ is the \emph{reverse triangle inequality} \cite[Prop. 2.5]{cavalletti2020} strongly reminiscent of \eqref{Eq:Reverse tau}: for \emph{every} $\mu,\nu,\sigma\in\scrP(\mms)$,
\begin{align}\label{Eq:Reverse triangle lp}
\ell_p(\mu,\sigma) \geq \ell_p(\mu,\nu) + \ell_p(\nu,\sigma).
\end{align}

\subsubsection{Timelike $p$-dualizability}\label{Sub:Timelike dual} Next, we review the concept of \textit{\textnormal{(}strong\textnormal{)} timelike $p$-dualizability}, $p\in(0,1]$, of pairs $(\mu,\nu)\in\scrP(\mms)$. It originates in \cite{cavalletti2020} and generalizes the notion of \emph{$p$-separation} from \cite[Def.~4.1]{mccann2020}. Pairs satisfying this condition allow for a good duality theory \cite[Prop.~2.19, Prop.~2.21, Thm.~2.26]{cavalletti2020}, which has been used to characterize $\smash{\ell_p}$-geodesics in the smooth case \cite[Thm.~4.3, Thm.~5.8]{mccann2020}. In our case, it is needed to set up the timelike curvature-dimension condition from \autoref{Def:TCD} below.

In view of the subsequent \autoref{Def:TL DUAL} taken from \cite[Def.~2.18, Def.~2.27]{cavalletti2020}, we refer to \cite[Def.~2.6]{cavalletti2020} for the inherent definition of \emph{cyclical monotonicity} of a subset of $\smash{\mms_\leq^2}$ with respect to $l^p$, which generalizes the standard concept of cyclical monotonicity with respect to any given cost function \cite[Def.~5.1]{villani2009}. It will only be relevant in \autoref{Le:Strong} below.

As usual, given any $a,b\colon\mms\to\R$ we define the function $a\oplus b\colon\mms^2\to\R$ by $(a\oplus b)(x,y) := a(x) + b(y)$.

\begin{definition}\label{Def:TL DUAL} Given $p\in (0,1]$, a pair $(\mu,\nu)\in\scrP(\mms)$ is termed 
\begin{enumerate}[label=\textnormal{\alph*.}]
\item  \emph{timelike $p$-dua\-lizable by $\pi\in \Pi_\ll(\mu,\nu)$} if  $\pi$ is an $\smash{\ell_p}$-optimal coupling, and  there exist Borel functions $a,b\colon\mms\to\R$ with $a\oplus b\in\Ell^1(\mms^2,\mu\otimes\nu)$ and $l^p\leq a\oplus b$ on $\supp\mu\times\supp\nu$,
\item  \emph{strongly timelike $p$-dualizable by $\pi\in\Pi_\ll(\mu,\nu)$} provided $(\mu,\nu)$ is timelike $p$-dualizable by $\pi$, and there exists some $l^p$-cyclically monotone Borel set $\Gamma\subset \smash{\mms_\ll^2\cap(\supp\mu_0\times\supp\mu_1)}$ such that any given coupling $\sigma\in \Pi_\leq(\mu_0,\mu_1)$ is $\ell_p$-optimal if and only if $\sigma[\Gamma]=1$, and
\item \emph{timelike $p$-dualizable} if $(\mu,\nu)$ is timelike $p$-dualizable by some $\pi\in\Pi_\ll(\mu,\nu)$; analogously for strong timelike $p$-dualizability.
\end{enumerate}
Moreover, any $\pi$ as above is called \emph{timelike $p$-dualizing}. 
\end{definition}

In this framework, we define
\begin{align*}
\TD_p(\mms) &:= \{(\mu,\nu)\in\scrP(\mms)^2 : (\mu,\nu)\textnormal{ timelike }p\textnormal{-dualizable}\},\\
\STD_p(\mms) &:= \{(\mu,\nu)\in\scrP(\mms)^2 : (\mu,\nu)\textnormal{ strongly timelike }p\textnormal{-dualizable}\}
\end{align*}

\begin{remark}\label{Re:Always exists} It will be useful to keep in mind that by definition, \emph{every} $\smash{\ell_p}$-optimal coupling of a strongly timelike $p$-dualizable pair is concentrated on $\smash{\mms_\ll^2}$.
\end{remark}

\begin{example}\label{Ex:STD} Evidently, if $\mu,\nu\in\scrP_\comp(\mms)$, then $(\mu,\nu)$ is timelike $p$-dualizable if and only if there exists an $\ell_p$-optimal coupling $\smash{\pi\in\Pi_\leq(\mu,\nu)}$ concentrated on $\smash{\mms_\ll^2}$.

A relevant example (already in the smooth case, cf.~\cite[Lem.~4.4, Thm.~7.1]{mccann2020}) for the strong version is the following. If $\mu,\nu\in\scrP_\comp(\mms)$ on a locally causally closed, globally hyperbolic Lorentzian geodesic space $(\mms,\met,\ll,\leq,\tau)$ with $\smash{\supp\mu\times\supp\nu\subset\mms_\ll^2}$, then the pair $(\mu,\nu)$ is strongly timelike $p$-dualizable, $p\in (0,1]$ \cite[Cor.~2.29]{cavalletti2020}.
\end{example}

\subsubsection{Geodesics revisited}\label{Sub:Geodesics}  Given  $p\in (0,1]$, following \cite[Subsec.~2.3.6, App.~B]{braun+}, see also \cite[Def.~2.31]{cavalletti2020} and \cite[Def.~1.1]{mccann2020}, we now recall the nonsmooth notion of timelike proper-time parametrized $\ell_p$-geodesics.

Recall the continuous reparametrization map $\sfr$ for elements of $\TGeo(\mms)$ introduced in  \autoref{Sub:GEO}. For $\mu_0,\mu_1\in\scrP(\mms)$, we define
\begin{align*}
\OptTGeo_{\ell_p}(\mu_0,\mu_1) &:= \{\bdpi\in\scrP(\Geo(\mms)) : (\eval_0,\eval_1)_\push\bdpi \in\Pi_\ll(\mu_0,\mu_1)\\
&\qquad\qquad\text{is }\ell_p\textnormal{-optimal}\},\\
\OptTGeo_{\ell_p}^\tau(\mu_0,\mu_1) &:= \sfr_\push\OptTGeo_{\ell_p}(\mu_0,\mu_1).
\end{align*}
All elements of the latter class are concentrated on the set $\TGeo^\tau(\mms)$ from \eqref{Eq:TGeo tau}.

In the following definition, we say that $\bdpi\in\scrP(\Cont([0,1];\mms))$ \emph{represents} a curve $(\mu_t)_{t\in[0,1]}$ in $\scrP(\mms)$ if $\mu_t= (\eval_t)_\push\bdpi$ holds for every $t\in[0,1]$.

\begin{definition}\label{Def:LP Geo} A collection $(\mu_t)_{t\in [0,1]}$ of elements of $\scrP(\mms)$ is termed \emph{timelike proper-time parametrized $\ell_p$-geodesic} if it is represented by some element $\bdpi$ belonging to $\smash{\OptTGeo_{\ell_p}^\tau(\mu_0,\mu_1)}$. Any such $\bdpi$ will be called \emph{timelike $\ell_p$-optimal geodesic plan}.
\end{definition}

Note that every curve $(\mu_t)_{t\in[0,1]}$ as in \autoref{Def:LP Geo} obeys
\begin{align*}
\ell_p(\mu_s,\mu_t) = (t-s)\,\ell_p(\mu_0,\mu_1) > 0.
\end{align*}
Thus, $(\mu_t)_{t\in[0,1]}$ is an $\smash{\ell_p}$-geodesic in the sense of \cite[Def.~2.13]{cavalletti2020} and \cite[Def.~1.1]{mccann2020} provided $\smash{\ell_p(\mu_0,\mu_1)<\infty}$. By regularity and geodesy of $(\mms,\met,\ll,\leq,\tau)$ and a standard measurable selection argument, cf.~\autoref{Le:Plans} and \autoref{Ass:Ass} below, timelike proper-time parametrized $\smash{\ell_p}$-geo\-de\-sics exist in great generality.

We then have the following result from \cite[Prop.~B.11]{braun+} used at many occasions below. The compactness of $\mms$ assumed therein is only made for notational simplicity and is not restrictive, as \autoref{Le:Plans} will always be applied within causal diamonds which are compact by assumption. Note that the chronology  assumption on the limit marginals in the last clause therein is essential. Indeed, unlike $\smash{\ell_p}$-optimality \cite[Sec.~2.3]{cavalletti2020}, chronology is in general \emph{not} stable under weak limits (in contrast to causality, which will be a closed condition by assumption).

To formulate the lemma, given $s,t\in[0,1]$ with $s<t$, let $\smash{\Restr_s^t}\colon\Cont([0,1];\mms) \to \Cont([0,1];\mms)$ be the restriction map defined by
\begin{align*}
\Restr_s^t(\gamma)_r := \gamma_{(1-r)s+rt}.
\end{align*}

\begin{lemma}\label{Le:Plans} Let $p\in (0,1]$ and suppose that $(\mms,\met,\ll,\leq,\tau)$ is a compact, causally closed, $\scrK$-globally hyperbolic, regular Lorentzian geodesic space. Suppose  $(\mu_0,\mu_1)\in \TD_p(\mms)$. Then the following properties hold.
\begin{enumerate}[label=\textnormal{\textcolor{black}{(}\roman*\textcolor{black}{)}}]
\item For every $\smash{\ell_p}$-optimal $\pi\in\Pi_\ll(\mu_0,\mu_1)$, there is $\bdpi\in\OptTGeo^\tau_{\ell_p}(\mu_0,\mu_1)$ such that $\smash{\pi = (\eval_0,\eval_1)_\push\bdpi}$.
\item There is at least one proper-time parametrized $\smash{\ell_p}$-geodesic from $\mu_0$ to $\mu_1$.
\item For every $\bdpi\in\smash{\OptTGeo_{\ell_p}^\tau(\mu_0,\mu_1)}$ and every $s,t\in[0,1]$ with $s<t$,
\begin{align*}
(\Restr_s^t)_\push\bdpi \in\OptTGeo_{\ell_p}^\tau((\eval_s)_\push\bdpi,(\eval_t)_\push\bdpi).
\end{align*}
\item If $\smash{\bdpi\in\OptTGeo_{\ell_p}^\tau(\mu_0,\mu_1)}$ and if $\smash{\bdsigma}$ is any nontrivial measure on  $\smash{\Cont([0,1];\mms)}$ with $\bdsigma\leq\bdpi$, then $\bdsigma[\Cont([0,1];\mms)]^{-1}\,\bdsigma$ is an element of $\smash{\OptTGeo_{\ell_p}^\tau(\sigma_0,\sigma_1)}$, where $\sigma_i := \bdsigma[\Cont([0,1];\mms)]^{-1}\,(\eval_i)_\push\bdsigma\in\scrP(\mms)$, $i\in\{0,1\}$.
\item If $(\mu_0,\mu_1)\in\STD_p(\mms)$ is the weak limit of a given sequence $\smash{(\mu_0^n,\mu_1^n)_{n\in\N}}$ in $\scrP(\mms)^2$, then every sequence $(\bdpi^n)_{n\in\N}$ satisfying $\smash{\bdpi^n\in\OptTGeo_{\ell_p}^\tau(\mu_0^n,\mu_1^n)}$ for every $n\in\N$ has an accumulation point, and any such point belongs to $\smash{\OptTGeo_{\ell_p}^\tau(\mu_0,\mu_1)}$.
\end{enumerate}
\end{lemma}

\subsection{Entropic timelike curvature-dimension condition}

\begin{definition} A sextuple $(\mms,\met,\meas,\ll,\leq,\tau)$ consisting of a Lorentzian pre-length space $(\mms,\met,\ll,\leq,\tau)$ endowed with a Radon measure $\meas$ as hypothesized in \autoref{Sub:Basic} will be called \emph{measured Lorentzian pre-length space}.
\end{definition}

For measured Lorentzian pre-length spaces, all notions from \autoref{Sec:Lorentzian geo} are understood with respect to the inherent Lorentzian pre-length structure.

\subsubsection{Timelike $(K,N,p)$-convexity} For later convenience, we introduce the following \autoref{Def:Convex} leaned on \cite[Def.~6.5]{mccann2020}. With a slight abuse of notation compared to \eqref{Eq:UN def}, given a functional $\rmE\colon \scrP(\mms)\to [-\infty,\infty]$ with sufficiently large finiteness domain $\Dom(\rmE)\subset\scrP(\mms)$ and $N\in (0,\infty)$, define $\scrU_N\colon \scrP(\mms) \to [0,\infty]$ by $\smash{\scrU_N(\mu) := \rme^{-\rmE(\mu)/N}}$. In our work,  the most relevant functional $\rmE$ is the relative entropy $\Ent_\meas$ introduced in \autoref{Sub:Entropy} below, but see also \autoref{Sub:Conv}.

For  $K\in\R$, $r\in[0,1]$, and $\vartheta \in [0,\infty]$, we consider the distortion coefficients
\begin{align*}
\sigma_{K,N}^{(r)}(\vartheta) := \begin{cases}
\displaystyle\frac{\sin(r\vartheta\sqrt{K/N})}{\sin(\vartheta\sqrt{K/N})} & \textnormal{if }0 <K\vartheta^2 <N\pi^2,\\
r & \textnormal{if }K\vartheta^2 =0,\\
\displaystyle\frac{\sinh(r\vartheta\sqrt{-K/N})}{\sinh(\vartheta\sqrt{-K/N})} & \textnormal{if } K\vartheta^2 < 0,\\
\infty & \textnormal{if }K\vartheta^2\geq \pi^2.
\end{cases}
\end{align*}
Here we employ the convention $0\cdot \infty := 0$.

\begin{definition}\label{Def:Convex} Let $p\in (0,1]$, $K\in\R$, and $N\in (0,\infty)$. We say that a functional $\rmE\colon\scrP(\mms)\to [-\infty,\infty]$ on a Lorentzian pre-length space $(\mms,\met,\ll,\leq,\tau)$ is \emph{$(K,N,p)$-convex} relative to $\scrQ\subset\scrP(\mms)^2$ if the following holds. For every $\mu_0,\mu_1\in\scrQ \cap \Dom(\rmE)^2$ with $\smash{\ell_p(\mu_0,\mu_1) < \infty}$ there exist an $\smash{\ell_p}$-optimal coupling $\smash{\pi\in \Pi_\ll(\mu_0,\mu_1)}$ and a timelike proper-time parametrized $\smash{\ell_p}$-geodesic $(\mu_t)_{t\in [0,1]}$ from $\mu_0$ to $\mu_1$ such that for every $t\in[0,1]$, we have
\begin{align*}
\scrU_N(\mu_t) \geq \sigma_{K,N}^{(1-t)}(T_\pi)\,\scrU_N(\mu_0) + \sigma_{K,N}^{(t)}(T_\pi)\,\scrU_N(\mu_1).
\end{align*}
\end{definition}

\begin{remark}\label{Re:Pathological} Unlike the metric definition of $(K,N)$-convex functions \cite[Def.~2.7]{erbar2015}, in \autoref{Def:Convex} the pathological situation $\smash{T_\pi=\infty}$ might occur. This either reduces to a trivial condition ($K<0$), does not involve any $\Ell^2$-norm of $\tau$ at all ($K=0$) or --- for $K>0$ and in the relevant case  $\rmE = \Ent_\meas$ --- cannot hold by the timelike  Bonnet--Myers theorem  \cite[Prop.~3.6]{cavalletti2020}.
\end{remark}

\begin{remark} In the framework of \autoref{Def:Convex}, define $e\colon[0,1]\to [-\infty,\infty]$ by $e(t) := \rmE(\mu_t)$. Then $(K,N,p)$-convexity of $\rmE$ is equivalent to the following. For every $\mu_0,\mu_1 \in\scrQ\cap\Dom(\rmE)^2$ with $\smash{\ell_p(\mu_0,\mu_1)<\infty}$ there exist an $\smash{\ell_p}$-optimal coupling $\pi\in\Pi_\ll(\mu_0,\mu_1)$ and a timelike proper-time parametrized $\smash{\ell_p}$-geodesic $(\mu_t)_{t\in[0,1]}$ such that if $\smash{e^{-1}(\{-\infty\})}$ is empty and $\smash{\Vert\tau\Vert_{\Ell^2(\mms^2,\pi)}<\infty}$, then $e$ is semiconvex on $(0,1)$ and satisfies
\begin{align*}
\ddot{e} - \frac{1}{N}\,\dot{e}^2 \geq K\,\big\Vert\tau\big\Vert_{\Ell^2(\mms^2,\pi)}^2
\end{align*}
in the distributional sense on $(0,1)$.
\end{remark}

\subsubsection{Relative entropy}\label{Sub:Entropy} %The relevant timelike curvature-dimension condition from \autoref{Def:TCD} is formulated in terms of convexity properties, according to \autoref{Def:Convex}, of the \emph{relative entropy} $\Ent$ with respect to the  reference measure $\meas$. The latter is shortly reviewed now, we refer to \cite{cavalletti2020, mccann2020, sturm2006a} for details.

We define $\Ent_\meas\colon\scrP(\mms) \to [-\infty,\infty]$ by
\begin{align*}
\Ent_\meas(\mu) = \begin{cases} \displaystyle \int_\mms \rho\log\rho\d\meas &\textnormal{if }\mu = \rho\,\meas\ll \meas,\ (\rho\log\rho)^+\in\Ell^1(\mms,\meas),\\
\infty & \textnormal{otherwise}.
\end{cases}
\end{align*}
This functional possesses the following  properties, details of which can be found in \cite{cavalletti2020, mccann2020, sturm2006a}. By Jensen's inequality, $\Ent_\meas(\mu) \geq -\log\meas[\supp\mu]>-\infty$ for every $\mu\in\scrP_\comp(\mms)$. Moreover, $\Ent_\meas$ is weakly lower semicontinuous in the following form: if a sequence $(\mu_n)_{n\in\N}$ in $\scrP(\mms)$ converges weakly to $\mu\in\scrP(\mms)$ and there is a Borel set $C\subset\mms$ with $\meas[C]<\infty$ and $\supp\mu_n\subset C$ for every $n\in\N$,  then
\begin{align*}
\Ent_\meas(\mu) \leq \liminf_{n\to\infty}\Ent_\meas(\mu_n).
\end{align*}

\subsubsection{The curvature-dimension condition} Now we finally come to the main definition from \cite{cavalletti2020}, namely \cite[Def.~3.2]{cavalletti2020}, based on the groundbreaking results \cite[Cor.~6.6, Cor.~7.5]{mccann2020} and \cite[Cor.~4.4]{mondinosuhr2018}. 

Recall the definition \eqref{Eq:UN def} of the exponentiated relative entropy $\scrU_N$.

\begin{definition}\label{Def:TCD} Let $(\mms,\met,\meas,\ll,\leq,\tau)$ be a measured Lorentzian pre-length space, and let $p\in (0,1)$, $K\in\R$, and $N\in (0,\infty)$. We say that the former satisfies
\begin{enumerate}[label=\textnormal{\alph*.}]
\item  the \emph{\textnormal{(}entropic\textnormal{)} timelike curvature-dimension condition} $\smash{\TCD_p^e(K,N)}$ if   $\scrU_N$ is $(K,N,p)$-convex relative to $\TD_p(\mms)$, and
\item the \emph{weak \text{(}entropic\text{)} timelike curvature-dimension condition} $\smash{\wTCD_p^e(K,N)}$ if $\scrU_N$ is $(K,N,p)$-convex relative to $\STD_p(\mms)\cap\scrP_\comp(\mms)^2$.
\end{enumerate}
\end{definition}

\begin{remark} If the space $(\mms,\met,\meas,\ll,\leq,\tau)$ is $\scrK$-globally hyperbolic and satisfies the $\smash{\wTCD_p(K,N)}$ condition (in fact, $\smash{\TMCP_p^e(K,N)}$ according to \autoref{Def:TMCP}  suffices), then it is timelike geodesic. If in addition, it is causally path connected \cite[Def.~3.4]{kunzinger2018} --- in particular, if $(\mms,\met,\meas,\ll,\leq,\tau)$ is a Lorentzian length space --- then it is geodesic \cite[Rem.~3.9]{cavalletti2020}. Hence, we may and will always assume the geodesic property with no  restriction.
\end{remark}

In \autoref{Sub:Infinite dim}, we  introduce an ``infinite-dimensional'' analogue of the $\wTCD$ condition in the spirit of \cite{sturm2006a}.

%\autoref{Def:TCD} has been proposed in \cite[Def.~3.2]{cavalletti2020} based on the groundbreaking works \cite{mccann2020, mondinosuhr2018} in the smooth Lorentzian framework, and provides a \emph{synthetic} meaning to ``the Ricci curvature and the dimension of of $(\mms,\met,\meas,\ll,\leq,\tau)$ being bounded from below by $K$ in all timelike directions and bounded from above by $N$, respectively''. Indeed, for a smooth Lorentzian spacetime, the latter statement is satisfied \emph{analytically} (i.e.~in terms of the given Ricci tensor and manifold dimension) if and only if $\smash{\TCD_p^e(K,N)}$, equivalently $\smash{\wTCD_p^e(K,N)}$, holds for some, hence every $p\in (0,1)$, cf.~\cite[Cor.~6.6, Cor.~7.5]{mccann2020} and \cite[Cor.~4.4]{mondinosuhr2018}. Here, the word ``entropic'' refers to an analogous curvature-dimension condition for metric measure spaces, originating in  \cite[Def.~3.1]{erbar2015}.

Among the many properties of these $\TCD$ and $\wTCD$ conditions proven in \cite{cavalletti2020}, let us quote: the timelike Brunn--Minkowski inequality \cite[Prop.~3.4]{cavalletti2020}, the timelike Bishop-Gromov inequality \cite[Prop.~3.5]{cavalletti2020}, the timelike Bonnet--Myers inequality \cite[Prop.~3.6]{cavalletti2020}, consistency and scaling properties \cite[Lem.~3.10]{cavalletti2020}, or  nonsmooth Hawking--Penrose singularity theorems \cite[Thm.~5.6, Cor.~5.13]{cavalletti2020}. The limit of a sequence of measured Lorentzian geodesic $\smash{\TCD_p^e(K,N)}$ spaces converging weakly, in a certain sense, is (only) $\smash{\wTCD_p^e(K,N)}$ \cite[Thm.~3.12]{cavalletti2020}. Lastly, under the additional assumption of timelike nonbranching, the following hold. Given $\mu_0\in\Dom(\Ent_\meas)$ and $\mu_1\in\scrP(\mms)$ admitting an $\smash{\ell_p}$-optimal coupling in $\smash{\Pi_\ll(\mu_0,\mu_1)}$, we have uniqueness of $\smash{\ell_p}$-optimal couplings of $\mu_0$ to $\mu_1$ \cite[Thm.~3.19]{cavalletti2020}; similarly, they are connected by a unique timelike proper-time parametrized $\smash{\ell_p}$-geodesic \cite[Thm.~3.20]{cavalletti2020}.

\begin{remark} Except for the Bonnet--Myers inequality, all preceding results are in fact valid under the weaker  \cite[Prop.~3.11]{cavalletti2020} timelike measure contraction property from \autoref{Def:TMCP} below.
\end{remark}

\section{Existence of good geodesics}\label{Ch:Good}

In this chapter, we prove \autoref{Th:Linfty estimates}. We show every intermediate result under the most general assumptions, possibly beyond those of \autoref{Th:Linfty estimates}. Together, however, these reduce precisely to the hypotheses of our main result.

\subsection{Strong timelike $p$-dualizability along $\smash{\ell_p}$-geodesics} The main argument for the construction of the timelike proper-time parametrized $\smash{\ell_p}$-geodesic for \autoref{Th:Linfty estimates} is based on bisection by iteratively selecting appropriate midpoints of timelike proper-time parametrized $\smash{\ell_p}$-geodesics. To this aim, we have to ensure that strong timelike $p$-dualizability behaves well along these. 

The proof of the corresponding nontrivial \autoref{Le:Strong} is grounded on a private communication of the author with Fabio Cavalletti and Andrea Mondino.

\begin{lemma}\label{Le:Strong} Let $(\mms,\met,\ll,\leq,\tau)$ be a globally hyperbolic, regular Lorentzian geodesic space, $p\in (0,1]$, and $(\mu_0,\mu_1)\in\scrP(\mms)^2$. Moreover, let $\smash{\bdpi\in\OptTGeo_{\ell_p}^\tau(\mu_0,\mu_1)}$ and define $\mu_t := (\eval_t)_\push\bdpi\in\scrP(\mms)$, $t\in [0,1]$. If the pair $(\mu_0,\mu_1)$ is \textnormal{(}strongly\textnormal{)} timelike $p$-dualizable, so is $(\mu_s,\mu_t)$ for every $s,t\in [0,1]$ with $s<t$. 
\end{lemma}

\begin{proof} We assume that $s,t\in(0,1)$, the case $\{s,t\}\cap\{0,1\}\neq \emptyset$ is analogous. Note that $(\eval_0,\eval_1)_\push\bdpi$ is concentrated on $\smash{\mms_\ll^2}$, and so is $\smash{(\eval_s,\eval_t)_\push\bdpi\in\Pi(\mu_s,\mu_t)}$. Since the latter is $\smash{\ell_p}$-optimal and the total cost
\begin{align*}
\ell_p(\mu_s,\mu_t) = (t-s)\,\ell_p(\mu_0,\mu_1)
\end{align*}
is positive and finite, the pair $(\mu_s,\mu_t)$ is timelike $p$-dualizable. 

To show strong timelike $p$-dualizability of $(\mu_s,\mu_t)$ if $(\mu_0,\mu_1)$ has this property, we have to construct an $l^p$-cyclically monotone Borel set $\Gamma_{s,t}\subset\smash{\mms_\ll^2\cap(\supp\mu_s\times\supp\mu_t)}$ such that $\pi[\Gamma_{s,t}]=1$ for every $\smash{\ell_p}$-optimal coupling $\smash{\pi\in\Pi_\leq(\mu_s,\mu_t)}$. To this aim, let $\Gamma\subset\smash{\mms_\ll^2\cap(\supp\mu_0\times\supp\mu_1)}$ be an $l^p$-cyclically monotone Borel  set on which every $\smash{\ell_p}$-optimal coupling belonging to $\Pi_\leq(\mu_0,\mu_1)$ is concentrated, and define
\begin{align*}
\Gamma_{s,t} := (\eval_s,\eval_t)\big[(\eval_0,\eval_1)^{-1}(\Gamma)\big].
\end{align*}
To show that $\Gamma_{s,t}$ is $l^p$-cyclically monotone, we  follow the proof of \cite[Lem.~4.4]{galaz2018}. Let $n\in\N$ and $(x^1, y^1),\dots,(x^n,y^n)\in \Gamma_{s,t}$, and select $\smash{\gamma^1,\dots,\gamma^n\in\TGeo^\tau(\mms)}$ with $\smash{(x^i,y^i) = (\gamma_s^i,\gamma_t^i)}$ for every $i\in\{1,\dots,n\}$. Since $\Gamma$ is $l^p$-cyclically monotone and since $\smash{\gamma^i\in\TGeo^\tau(\mms)}$  for every $i\in\{1,\dots,n\}$, the empirical measure $\bdsigma$ of $\smash{\gamma^1,\dots,\gamma^n}$ 
is a timelike $\smash{\ell_p}$-optimal geodesic plan interpolating its endpoints \cite[Prop.~2.8]{cavalletti2020}. Therefore, $(\eval_s,\eval_t)_\push\bdsigma$ is an $\smash{\ell_p}$-optimal coupling of its marginals, and applying \cite[Prop.~2.8]{cavalletti2020} again yields the $l^p$-cyclical monotonicity of
\begin{align*}
\supp (\eval_s,\eval_t)_\push\bdsigma = \bigcup_{i=1}^n \big\lbrace (\gamma_s^i,\gamma_t^i)\big\rbrace = \bigcup_{i=1}^n \big\lbrace (x^i,y^i)\big\rbrace.
\end{align*}

Given any $\smash{\ell_p}$-optimal coupling $\smash{\pi\in\Pi_\leq(\mu_s,\mu_t)}$, by gluing and a measurable se\-lection argument as in the proof of \cite[Thm.~2.10]{ambrosiogigli}, using that $\mu_s$ and $\mu_t$ lie on a timelike proper-time parametrized $\smash{\ell_p}$-geodesic, we find $\bdalpha\in\smash{\OptTGeo_{\ell_p}^\tau(\mu_0,\mu_1)}$  with $(\eval_s,\eval_t)_\push\bdalpha = \pi$. Noting that
\begin{align*}
\pi[\Gamma_{s,t}] = (\eval_s,\eval_t)_\push\bdalpha[\Gamma_{s,t}] = (\eval_0,\eval_1)_\push\bdalpha[\Gamma] = 1
\end{align*}
then terminates the proof.
\end{proof}

%\begin{remark}\label{Re:No weak?} An analogue to \autoref{Le:Strong} for timelike $p$-dualizability is unclear to us, at least if timelike branching occurs (by uniqueness of $\smash{\ell_p}$-optimal couplings in the nonbranching situation, see \cite[Thm.~3.19, Thm.~3.20]{cavalletti2020} and \autoref{Cor:Timelike branching}). The problem is that for $\smash{\bdpi\in\OptTGeo_{\ell_p}^\tau(\mu_0,\mu_1)}$ interpolating a timelike $p$-dualizable pair  $(\mu_0,\mu_1)\in\scrP(\mms)$, despite the existence of a chronological $\smash{\ell_p}$-optimal coupling of $\mu_0$ and $\mu_1$ it might happen that $\smash{(\eval_0,\eval_1)_\push\bdpi}$ gives mass to $\{\tau=0\}$. Thus, $(\eval_s,\eval_t)_\push\bdpi$ is not any more a candidate for a chronological coupling of $(\eval_s)_\push\bdpi$ and $(\eval_t)_\push\bdpi$ in Step 1 of the above proof of \autoref{Le:Strong}. 

%\emph{Assuming} that $\smash{(\eval_0,\eval_1)_\push\bdpi[\{\tau = 0\}] = 0}$ would not be enough for our purposes either: some of our arguments (see e.g.~\autoref{Le:Maximizer} or \autoref{Cor:Existence minimizer}) rely on weak compactness of $\smash{\OptTGeo_{\call_p}^\tau(\mu_0,\mu_1)}$, but the previous condition does generally not pass over to weak limits. This is the main reason why we were unable to prove the $\TCD$ counterpart of \autoref{Th:Linfty estimates} beyond strictly timelike $p$-ordered pairs of measures, as described in \autoref{Re:Minor} below. On the other hand, by  \autoref{Re:Always exists} this stability is no issue in the strong timelike $p$-dualizable case.
%\end{remark}

\subsection{Construction of a candidate}\label{Sub:Construction} In this section, we construct an appropriate timelike proper-time parametrized $\smash{\ell_p}$-geodesic $(\mu_t)_{t\in [0,1]}$ for which we verify in \autoref{Sub:Verification} and \autoref{Sub:Uniform density bounds} that it satisfies the goodness properties from \autoref{Def:Good}.

\begin{assumption}\label{Ass:Ass} From now on, until the end of this article, and  unless explicitly stated otherwise we assume $(\mms,\met,\meas,\ll,\leq,\tau)$ to be a causally closed, $\scrK$-globally hyperbolic, regular Lorentzian geodesic space.
\end{assumption}

Given $N\in (0,\infty)$, let $\scrU_N$ be as in \eqref{Eq:UN def}, and for $t\in (0,1)$ define the functional $\scrV_N^t \colon\scrP(\Cont([0,1];\mms)) \to [0,\infty]$ through
\begin{align*}
\scrV_N^t(\bdpi) := \scrU_N((\eval_t)_\push\bdpi).
\end{align*}
Except for \autoref{Sub:TMCP cond}, we  mostly work with the functional
\begin{align}\label{Eq:scrV}
\scrV_N := \scrV_N^{1/2}.
\end{align}

\begin{remark}  $\scrV_N$ only depends on a single slice of its argument, and one is tempted to follow the $\CD$-treatise  \cite{rajala2012a} more closely instead and consider the functional $\scrU_N$ on the set of $\ell_p$-midpoints of $\mu_0$ and $\mu_1$ similar to \cite[Ch.~3]{rajala2012a} or \cite[Sec.~3.2]{rajala2012b}. However, in our case it is more convenient to work with timelike $\smash{\ell_p}$-optimal geodesic plans. For instance, for $\smash{\bdpi\in\OptTGeo_{\ell_p}^\tau(\mu_0,\mu_1)}$, $\mu_0,\mu_1\in\scrP(\mms)$, the pairs $(\mu_0,(\eval_{1/2})_\push\bdpi)$ and $((\eval_{1/2})_\push\bdpi,\mu_1)$ inherit the dualizability and chronology  properties of $(\mu_0,\mu_1)$ --- needed e.g.~for  \autoref{Pr:Min 0}  below (and recall  \autoref{Le:Strong}) --- while this seems  unclear for general $\smash{\ell_p}$-mid\-points.
\end{remark}

\begin{lemma}\label{Le:Maximizer} Let $p,t\in (0,1)$, $N\in (0,\infty)$, as well as $\mu_0,\mu_1\in\STD_p(\mms)\cap\scrP_\comp(\mms)^2$. Then $\smash{\scrV_N^t}$ has a maximizer in $\smash{\OptTGeo_{\ell_p}^\tau(\mu_0,\mu_1)}$ with finite value. Moreover, if $\smash{\wTCD_p^e(K,N)}$ holds for some $K\in\R$ and $N\in (0,\infty)$, and if the pair $(\mu_0,\mu_1)\in(\scrP_\comp(\mms)\cap\Dom(\Ent_\meas))^2$ is strongly timelike $p$-dualizable, for every  maximizer $\bdpi\in\smash{\OptTGeo_{\ell_p}^\tau(\mu_0,\mu_1)}$ of $\smash{\scrV_N^t}$ the measure $(\eval_t)_\push\bdpi\in\scrP_\comp(\mms)$ has finite entropy; in particular $(\eval_t)_\push\bdpi\ll\meas$.
\end{lemma}

\begin{proof} First, recall from \autoref{Sub:Lorentz} that $\smash{\OptTGeo_{\ell_p}^\tau(\mu_0,\mu_1)\neq\emptyset}$. As $\supp(\eval_t)_\push\bdpi = \{\gamma_t : \gamma\in\supp\bdpi\} \subset J(\mu_0,\mu_1)$ for every $\bdpi\in \smash{\OptTGeo_{\ell_p}^\tau(\mu_0,\mu_1)}$,  we have
\begin{align*}
\scrV_N^t(\bdpi) \leq \meas\big[J(\mu_0,\mu_1)\big]^{1/N}
\end{align*}
by Jensen's inequality. Thus, $\smash{\scrV_N^t}$ is bounded on $\smash{\OptTGeo_{\ell_p}^\tau(\mu_0,\mu_1)}$.  

Moreover, $\smash{\scrV_N^t}$ is weakly upper semicontinuous on $\smash{\OptTGeo_{\ell_p}^\tau(\mu_0,\mu_1)}$. Since the latter is weakly compact by \autoref{Le:Plans}, the existence of a maximizer for $\smash{\scrV_N^t}$ follows from the direct method.

The last claim follows by taking the $t$-slice of a timelike $\smash{\ell_p}$-optimal geodesic plan representing a timelike proper-time parametrized $\smash{\ell_p}$-geodesic from $\mu_0$ to $\mu_1$ witnessing the $(K,N,p)$-convexity inequality of $\scrU_N$ as a competitor. Hence, the maximum of $\scrV_N$ is strictly positive.
\end{proof}

We construct $(\mu_t)_{t\in [0,1]}$ as follows. Let the pair  $\mu_0,\mu_1\in\scrP_\comp(\mms)\cap\Dom(\Ent_\meas)$ be strongly time\-like $p$-dualizable. Initially, set $\mu_{1/2} := (\eval_{1/2})_\push\bdpi_1\in\scrP_\comp(\mms)\cap\Dom(\Ent_\meas)$, where $\bdpi_1\in\smash{\OptTGeo_{\ell_p}^\tau(\mu_0,\mu_1)}$ is a maximizer of $\scrV_N$ according to \autoref{Le:Maximizer}. 

By induction, suppose that for a given $n\in\N$ we have defined $\mu_{k2^{-n}}\in\scrP_\comp(\mms)\in\Dom(\Ent_\meas)$ for every $k\in\{0,\dots,2^n\}$. For every odd $\smash{k\in \{1,\dots,2^{n+1}-1\}}$, by construction the pair $(\mu_{(k-1)2^{-n-1}},\mu_{(k+1)2^{-n-1}})\in(\scrP_\comp(\mms)\cap\Dom(\Ent_\meas))^2$  is strongly timelike $p$-dualizable thanks to \autoref{Le:Strong}. Let
\begin{align*}
\bdpi^k_{n+1}\in\OptTGeo_{\ell_p}^\tau(\mu_{(k-1)2^{-n-1},(k+1)2^{-n-1}})
\end{align*}
maximize $\scrV_N$ on the latter set, cf.~\autoref{Le:Maximizer}. We glue together these timelike $\smash{\ell_p}$-optimal geodesic plans $\smash{\bdpi_{n+1}^0,\dots,\bdpi_{n+1}^{2^n}}$ and obtain $\smash{\bdpi_{n+1}\in\OptTGeo_{\ell_p}^\tau(\mu_0,\mu_1)}$. Inductively, we thus get a sequence $\smash{(\bdpi_n)_{n\in\N}}$ in $\smash{\OptTGeo_{\ell_p}^\tau(\mu_0,\mu_1)}$ which, by  \autoref{Le:Plans}, has a weak limit $\smash{\bdpi\in\OptTGeo_{\ell_p}^\tau(\mu_0,\mu_1)}$ along a nonrelabeled subsequence. In turn, the plan $\bdpi$ induces a timelike proper-time parametrized $\smash{\ell_p}$-geodesic $(\mu_t)_{t\in[0,1]}$ defined by
\begin{align*}
\mu_t = (\eval_t)_\push\bdpi.
\end{align*}

In the rest of this chapter, we refer to $(\mu_t)_{t\in[0,1]}$  as the \emph{candidate} (but we may use the notation $(\mu_t)_{t\in[0,1]}$ at other occasions as well, whenever convenient).

Let $\boldsymbol{\mathrm{D}}\subset\Q$ henceforth denote the set of dyadic numbers.

\subsection{$(K,N,p)$-convexity inequality}\label{Sub:Verification} Now we start proving that the candidate $(\mu_t)_{t\in[0,1]}$ is good according to \autoref{Def:Good}: it has to obey the $(K,N,p)$-convexity inequality for $\scrU_N$ defining $\smash{\wTCD_p^e(K,N)}$ for  $p\in (0,1)$, $K\in\R$, and $N\in (0,\infty)$, and  $\Vert \rho_t\Vert_{\Ell^\infty(\mms,\meas)}$ has to be uniformly bounded in $t\in[0,1]$ subject to the decomposition $\mu_t = \rho_t\,\meas$. (Recall that $\mu_t\in\Dom(\Ent_\meas)$ for every $t\in[0,1]$ by \autoref{Le:Maximizer}, weak lower semicontinuity of $\Ent_\meas$, and Jensen's inequality.) We start with the former.

The subsequent identities taken from \cite[Lem. 3.2]{rajala2012a} are crucial in the proof of the main  \autoref{Pr:Semic}. 

\begin{lemma}\label{Le:Distortion} Let $K\in\R$ and $N\in (0,\infty)$, and let  $t_1,t_2,t_3\in[0,1]$ with $t_1 < t_2$ as well as $\vartheta\geq 0$. Then
\begin{align*}
\sigma_{K,N}^{((1-t_3)t_1 + t_3t_2)}(\vartheta) &= \sigma_{K,N}^{(1-t_3)}((t_2-t_1)\vartheta)\,\sigma_{K,N}^{(t_1)}(\vartheta)\\
&\qquad\qquad  + \sigma_{K,N}^{(t_3)}((t_2-t_1)\vartheta)\,\sigma_{K,N}^{(t_2)}(\vartheta),\\
\sigma_{K,N}^{(1-(1-t_3)t_1-t_3t_2)}(\vartheta) &= \sigma_{K,N}^{(1-t_3)}((t_2-t_1)\vartheta)\,\sigma_{K,N}^{(1-t_1)}(\vartheta)\\
&\qquad\qquad + \sigma_{K,N}^{(t_3)}((t_2-t_1)\vartheta)\,\sigma_{K,N}^{(1-t_2)}(\vartheta).
\end{align*}
\end{lemma}

\begin{remark} Recall that analogous formulas are not valid for the distortion coefficients used to set up the finite-dimensional $\CD$ condition for metric measure spaces \cite[Def.~1.3]{sturm2006b}. Related to this, \autoref{Le:Distortion} is  the main reason for the local-to-global property of its reduced counterpart \cite[Thm.~5.1]{bacher2010}, see also \cite[Thm.~3.45]{braun+}.
\end{remark}

\begin{proposition}\label{Pr:Semic} Assume $\smash{\wTCD_p^e(K,N)}$ for some $p\in (0,1)$, $K\in\R$ and $N\in (0,\infty)$. Let  $(\mu_0,\mu_1)\in\scrP_\comp(\mms)\cap\Dom(\Ent_\meas)$ be strongly timelike $p$-dualizable. Then there exists some timelike $p$-dualizing coupling $\smash{\pi\in\Pi_\ll(\mu_0,\mu_1)}$ such that the candidate $(\mu_t)_{t\in [0,1]}$ associated to $\mu_0$ and $\mu_1$ from \autoref{Sub:Construction} obeys, for every $t\in[0,1]$,
\begin{align*}
\scrU_N(\mu_t) \geq \sigma_{K,N}^{(1-t)}(T_\pi)\,\scrU_N(\mu_0) + \sigma_{K,N}^{(t)}(T_\pi)\,\scrU_N(\mu_1).
\end{align*}
\end{proposition}

\begin{proof} Given the above candidate, we have to construct $\pi$. To this aim, we loosely follow \cite[Sec.~3.1]{rajala2012a}, but have to perform a subtle modification. The curvature-dimension condition in \cite[Def.~1.1]{rajala2012a} might be different from its entropic counterpart for metric measure spaces which are not essentially nonbran\-ching \cite[Def.~3.1, Thm.~3.12]{erbar2015}, and the $\smash{\TCD_p^e(K,N)}$ condition from \autoref{Def:TCD} is formulated in the spirit of \cite{erbar2015}. In particular, the timelike proper-time parametrized $\smash{\ell_p}$-geodesic and the $\smash{\ell_p}$-optimal coupling of $\mu_0$ and $\mu_1$ therein have a priori nothing to do with each other, unlike  \cite[Def.~1.1]{rajala2012a}. We thus have to keep track of all couplings appearing in the $\TCD$ condition at every step of the dyadic construction from \autoref{Sub:Construction}. 

This is done by a monotonicity argument by gradually selecting  the plan with respect to which the $\Ell^2$-norm of $\tau$ is maximal if $K < 0$ or minimal if $K\geq 0$, and a tightness argument justifying the final passage to the limit. For simplicity, let us assume that $K<0$, the other case is treated analogously.

\textbf{Step 1.} \textit{Approximate $(K,N,p)$-convexity inequality for dyadic times.} By maximality of $\scrV_N$,  the $\smash{\wTCD_p^e(K,N)}$ condition, and \autoref{Le:Plans} there is an $\smash{\ell_p}$-optimal coupling $\pi_1\in\Pi_\ll(\mu_0,\mu_1)$ such that
\begin{align*}
\scrU_N(\mu_{1/2}) \geq \sigma_{K,N}^{(1/2)}(T_{\pi_1})\,\scrU_N(\mu_0) + \sigma_{K,N}^{(1/2)}(T_{\pi_1})\,\scrU_N(\mu_1).
\end{align*}

Now suppose that for every $n\in\N$ there exists $\smash{\pi_n\in\Pi_\ll(\mu_0,\mu_1)}$ such that for every $\smash{k\in\{1,\dots,2^n-1\}}$,
\begin{align}\label{Eq:PIN}
\scrU_N(\mu_{k2^{-n}}) \geq \sigma_{K,N}^{(1-k2^{-n})}(T_{\pi_n})\,\scrU_N(\mu_0) + \sigma_{K,N}^{(k2^{-n})}(T_{\pi_n})\,\scrU_N(\mu_1).
\end{align}
Let $k\in \{1,\dots,2^{n+1}-1\}$ be an odd number. Arguing as above and noting that the ancestors $\smash{\mu_{(k-1)2^{-n-1}}}$ and $\smash{\mu_{(k+1)2^{-n-1}}}$ of $\mu_{k2^{-n-1}}$ are strongly timelike $p$-dualizable, there is an $\smash{\ell_p}$-optimal coupling $\smash{\omega_{n+1}^k \in \Pi_\ll(\mu_{(k-1)2^{-n-1}},\mu_{(k+1)2^{-n-1}})}$ such that
\begin{align*}
\scrU_N(\mu_{k2^{-n-1}}) &\geq \sigma_{K,N}^{(1/2)}(T_{\omega^k_{n+1}})\,\scrU_N(\mu_{(k-1)2^{-n-1}})\\
&\qquad\qquad + \sigma_{K,N}^{(1/2)}(T_{\omega^k_{n+1}})\,\scrU_N(\mu_{(k+1)2^{-n-1}})\\
&\geq \sigma_{K,N}^{(1/2)}(T_{\omega_{n+1}^k})\,\sigma_{K,N}^{(1-(k-1)2^{-n-1})}(T_{\pi_n})\,\scrU_N(\mu_0)\\
&\qquad\qquad + \sigma_{K,N}^{(1/2)}(T_{\omega_{n+1}^k})\,\sigma_{K,N}^{((k-1)2^{-n-1})}(T_{\pi_n})\,\scrU_N(\mu_1)\\
&\qquad\qquad+ \sigma_{K,N}^{(1/2)}(T_{\omega_{n+1}^k})\,\sigma_{K,N}^{(1-(k+1)2^{-n-1})}(T_{\pi_n})\,\scrU_N(\mu_0)\\
&\qquad\qquad + \sigma_{K,N}^{(1/2)}(T_{\omega_{n+1}^k})\,\sigma_{K,N}^{((k+1)2^{-n-1})}(T_{\pi_n})\,\scrU_N(\mu_1).
\end{align*}
In the second inequality, we have used our induction hypothesis. By \autoref{Le:Plans} and arguing as for  \cite[Thm.~2.11]{ambrosiogigli} we now construct a plan $\smash{\bdalpha_{n+1}^k\in\OptTGeo_{\ell_p}^\tau(\mu_0,\mu_1)}$, which is henceforth fixed, with the property 
\begin{align*}
(\eval_{(k-1)2^{-n-1}},\eval_{(k+1)2^{-n-1}})_\push\bdalpha_{n+1}^k = \omega_{n+1}^k.
\end{align*}

Having at our disposal these  timelike $\smash{\ell_p}$-optimal geodesic plans $\smash{\bdalpha_{n+1}^k}$ for every odd index $k\in\{1,\dots,2^{n+1}-1\}$, employing that
\begin{align*}
T_{\omega_{n+1}^k} = 2^{-n}\,T_{\pi_{n+1}^k}
\end{align*}
for $\smash{\pi_{n+1}^k := (\eval_0,\eval_1)_\push\bdalpha_{n+1}^k}$,  considering the $\smash{\ell_p}$-optimal coupling
\begin{align*}
\pi_{n+1} := \argmax\!\big\lbrace T_\pi : \pi \in\{\pi_n,\pi_{n+1}^1,\dots,\pi_{n+1}^{2^{n+1}-1}\}\big\rbrace
\end{align*}
of $\mu_0$ and $\mu_1$, and that the function $\smash{\sigma_{K,N}^{(r)}(\vartheta)}$ is increasing in $\vartheta\geq 0$ for every $r\in[0,1]$, from the above inequalities we obtain
\begin{align*}
\scrU_N(\mu_{k2^{-n-1}}) &\geq \sigma_{K,N}^{(1/2)}(2^{-n}\,T_{\pi_{n+1}})\,\sigma_{K,N}^{(1-(k-1)2^{-n-1})}(T_{\pi_{n+1}})\,\scrU_N(\mu_0)\\
&\qquad\qquad + \sigma_{K,N}^{(1/2)}(2^{-n}\,T_{\pi_{n+1}})\,\sigma_{K,N}^{((k-1)2^{-n-1})}(T_{\pi_{n+1}})\,\scrU_N(\mu_1)\\
&\qquad\qquad+ \sigma_{K,N}^{(1/2)}(2^{-n}\,T_{\pi_{n+1}})\,\sigma_{K,N}^{(1-(k+1)2^{-n-1})}(T_{\pi_{n+1}})\,\scrU_N(\mu_0).\\
&\qquad\qquad + \sigma_{K,N}^{(1/2)}(2^{-n}\,T_{\pi_{n+1}})\,\sigma_{K,N}^{((k+1)2^{-n-1})}(T_{\pi_{n+1}})\,\scrU_N(\mu_1)\\
&= \sigma_{K,N}^{(1-k2^{-n-1})}(T_{\pi_{n+1}})\,\scrU_N(\mu_0) + \sigma_{K,N}^{(k2^{-n-1})}(T_{\pi_{n+1}})\,\scrU_N(\mu_1).
\end{align*}
In the last step, we have used \autoref{Le:Distortion}. 

\textbf{Step 2.} \textit{Construction of $\pi$ and conclusion.} By induction, we have thus obtained a sequence $(\pi_n)_{n\in\N}$ of $\smash{\ell_p}$-optimal couplings of $\mu_0$ and $\mu_1$ such that $\pi_n$ satisfies \eqref{Eq:PIN} for every $\smash{k\in\{1,\dots,2^n-1\}}$, $n\in\N$. Since $\supp\pi_n\subset \supp\mu_0\times\supp\mu_1$ is compact for every $n\in\N$, Prokhorov's theorem, stability \cite[Thm.~2.14]{cavalletti2020} and strong timelike $p$-dualizability of $\mu_0$ and $\mu_1$ imply weak convergence of a nonrelabeled  subsequence of $(\pi_n)_{n\in\N}$ to an $\smash{\ell_p}$-optimal coupling $\smash{\pi\in\Pi_\ll(\mu_0,\mu_1)}$. Since $\tau$ is continuous and bounded on $\supp\mu_0\times\supp\mu_1$, we have $\smash{T_{\pi_n}\to T_\pi}$ as $n\to\infty$. Sending $n\to\infty$ in the  inequality for $\scrU_N$ from Step 1 and  employing weak upper semicontinuity of $\scrU_N$ in the case $t\in [0,1]\setminus\boldsymbol{\mathrm{D}}$ thus gives the desired inequality.

\textbf{Step 3.} \textit{Properties of $\pi$.} By \cite[Thm.~2.14]{cavalletti2020}, $\pi$ constitutes in fact an $\smash{\ell_p}$-optimal coupling of $\mu_0$ and $\mu_1$. As such, it is concentrated on $\smash{\mms_\ll^2}$ thanks to  \autoref{Re:Always exists}, whence it is timelike $p$-dualizing.
\end{proof}

\subsection{Uniform density bounds}\label{Sub:Uniform density bounds} Now we show that the candidate $(\mu_t)_{t\in [0,1]}$ from \autoref{Sub:Construction} satisfies the desired $\Ell^\infty$-bounds for its densities with respect to $\meas$. This requires some preliminary work culminating in \autoref{Pr:Max = 0} below, where it turns out that maximizers of $\scrV_N$ directly  admit the correct density bounds.

\subsubsection{Spread of mass} First, we examine how the $\wTCD$ condition spreads mass along appropriate timelike proper-time parametrized $\smash{\ell}_p$-geodesics. In view of \autoref{Pr:Min 0}, \autoref{Cor:Thr}, and \autoref{Pr:Max = 0} this is the key result which provides us with the critical threshold for the $\Ell^\infty$-norm of $\rho_t$ subject to the decomposition $\mu_t= \rho_t\,\meas$, $t\in[0,1]$.

\begin{lemma}\label{Le:Bounded midpoints} Let $(\mms,\met,\meas,\ll,\leq,\tau)$  satisfy $\wTCD_p^e(K,N)$ for some $p\in (0,1)$, $K\in\R$ and $N\in (0,\infty)$. Suppose  $(\mu_0,\mu_1) = (\rho_0\,\meas,\rho_1\,\meas) \in\smash{\scrP_\comp^\ac(\mms,\meas)^2}$ is strongly timelike $p$-dualizable, and that $\rho_0,\rho_1\in\Ell^\infty(\mms,\meas)$, $i\in \{0,1\}$.  Then there exists a timelike proper-time parametrized $\smash{\ell_p}$-geodesic $(\mu_t)_{t\in [0,1]}$ connecting $\mu_0$ and $\mu_1$ such that $\mu_t=\rho_t\,\meas\in\Dom(\Ent_\meas)$ for every $t\in (0,1)$, and 
\begin{align*}
\meas\big[\{\rho_{1/2} > 0\}\big] \geq \rme^{-D\sqrt{K^-N}/2}\,\max\!\big\lbrace \Vert\rho_0\Vert_{\Ell^\infty(\mms,\meas)},\Vert\rho_1\Vert_{\Ell^\infty(\mms,\meas)}\big\rbrace^{-1},
\end{align*}
where  $D := \sup\tau(\supp\mu_0\times\supp\mu_1)$.
\end{lemma}

\begin{proof} First, note that  $\mu_0,\mu_1\in\Dom(\Ent_\meas)$. Moreover, $\sup\tau(\supp\mu_0\times\supp\mu_1) < \infty$ by $\scrK$-global hy\-per\-bolicity. Lastly, as $\smash{\wTCD_p^e(K,\infty)}$ implies $\smash{\wTCD_p^e(-K^-,\infty)}$, we may and will assume without restriction that $K\leq 0$.

Let $\smash{\pi\in\Pi_\ll(\mu_0,\mu_1)}$ be a timelike $p$-dua\-li\-zing coup\-ling for $(\mu_0,\mu_1)$ and $(\mu_t)_{t\in [0,1]}$ be a timelike proper-time parametrized $\ell_p$-geodesic from $\mu_0$ to $\mu_1$ along which $\scrU_N$ obeys the $(K,N,p)$-convexity pro\-per\-ty from \autoref{Def:TCD}. By \autoref{Le:Plans}, $(\mu_t)_{t\in [0,1]}$ is represented by some plan $\bdpi\in\smash{\OptTGeo_{\ell_p}^\tau(\mu_0,\mu_1)}$. Since $\supp \mu_t= \{ \gamma_t : \gamma\in \supp\bdpi \} \subset J(\mu_0,\mu_1)$ is compact, the $\TCD$-property  implies that $\mu_t\in \Dom(\Ent_\meas)$ for every given $t\in (0,1)$. Moreover, 
\begin{align*}
E := \{\rho_{1/2} > 0\}
\end{align*}
is contained in a compact set, whence $\meas[E] \in (0,\infty)$. Set
\begin{align*}
R := \max\!\big\lbrace \Vert \rho_0\Vert_{\Ell^\infty(\mms,\meas)}, \Vert \rho_1\Vert_{\Ell^\infty(\mms,\meas)}\big\rbrace
\end{align*}
and note that for every $\vartheta \in [0,\infty)$ and every given $t\in(0,1)$, we have
\begin{align*}
\sigma_{K,N}^{(t)}(\vartheta) \geq t\,\rme^{-(1-t)\vartheta\sqrt{-K/N}},
\end{align*}
see e.g.~\cite[Rem.~2.3]{cavalletti2017}. Therefore
\begin{align*}
\scrU_N(\mu_{1/2}) &\geq \sigma_{K,N}^{(1/2)}(T_\pi)\,\scrU_N(\mu_0) + \sigma_{K,N}^{(1/2)}(T_\pi)\,\scrU_N(\mu_1)\!\textcolor{white}{\Big\vert}\\
&\geq \frac{1}{2}\,\rme^{-D\sqrt{-K/N}/2}\,\scrU_N(\mu_0) +\frac{1}{2}\,\rme^{-D\sqrt{-K/N}/2}\, \scrU_N(\mu_1)\\
&\geq \rme^{-D\sqrt{-K/N}/2}\,R^{-1/N}.
\end{align*}
Here we used that $\pi\textnormal{-}\!\esssup \tau(\mms^2) \leq D$.  On the other hand,  $\smash{\scrU_N(\mu_{1/2}) \leq \meas[E]^{1/N}}$ as in the proof of \autoref{Le:Maximizer}. The claim follows.
\end{proof}

\begin{remark} Of course, the same reasoning yields
\begin{align*}
\meas\big[\{\rho_t > 0\}\big] \geq \rme^{-D\sqrt{K^-/N}}\,\max\!\big\lbrace\Vert\rho_0\Vert_{\Ell^\infty(\mms,\meas)}, \Vert\rho_1\Vert_{\Ell^\infty(\mms,\meas)}\big\rbrace^{-1}
\end{align*}
for every $t\in (0,1)$ in the situation of \autoref{Le:Bounded midpoints}. Note that for $t=1/2$, which is the relevant case in the sequel, \autoref{Le:Bounded midpoints} provides a better constant, though.
\end{remark}

%\begin{remark}\label{Re:Statement} Under the additional assumption that $I^-(x) \neq \emptyset\neq I^+(y)$ for every $x\in \supp\mu_0$ and every $y\in\supp\mu_1$, the condition of $\scrK$-global hyperbolicity can be weakened to global hyperbolicity \cite[Lem.~1.5, Rem.~2.33]{cavalletti2020}.
%\end{remark}

\subsubsection{Mass excess functional} Now we study the mass excess functional we deal with later especially in \autoref{Pr:Min 0}, \autoref{Cor:Thr}, and \autoref{Pr:Max = 0}. It has already been considered in \cite{rajala2012a,rajala2012b} in the context of metric measure spaces. It measures how much its input deviates from satisfying our density requirements for a good geodesic. Given $c\geq 0$ as well as $t\in (0,1)$, define $\smash{\scrF_c\colon\scrP(\mms)\to [0,1]}$ by
\begin{align}\label{Eq:F Functional}
\scrF_c(\mu) &:= \big\Vert (\rho - c)^+\big\Vert_{\Ell^1(\mms,\meas)} + \mu_\perp[\mms]
\end{align} 
subject to the decomposition $\mu = \rho\,\meas + \mu_\perp$, and $\smash{\scrE_c^t\colon \scrP(\Cont([0,1];\mms)) \to [0,1]}$ by
\begin{align*}
\scrE_c^t(\bdpi) &:= \scrF_c((\eval_t)_\push\bdpi).
\end{align*} 
Except for \autoref{Sub:TMCP cond} below, we mostly work with the functional
\begin{align}\label{Eq:Functional E}
\scrE_c := \scrE_c^{1/2}.
\end{align}

\begin{lemma}\label{Cor:Existence minimizer} Let $p\in (0,1]$, $t\in (0,1)$, as well as $c\geq 0$. Suppose that $(\mu_0,\mu_1)\in\STD_p(\mms)\cap\scrP_\comp(\mms)^2$. Then $\smash{\scrE_c^t}$ has a minimizer in $\smash{\OptTGeo_{\ell_p}^\tau(\mu_0,\mu_1)}$. 
\end{lemma}

\begin{proof} Since $J(\mu_0,\mu_1)$ is compact, the functional  $\scrF_c$ is weakly lower semicontinuous on $\scrP(J(\mu_0,\mu_1))$, cf.~e.g.~\cite[Thm.~30.6]{villani2009} or \cite[Lem.~3.6]{rajala2012b}. Hence, $\smash{\scrE_c^t}$ is weakly lower semicontinuous on $\smash{\OptTGeo_{\ell_p}^\tau(\mu_0,\mu_1)}$. The claim follows as for \autoref{Le:Maximizer}.
\end{proof}

\subsubsection{$\Ell^\infty$-bounds for minimizers of $\scrE_c$} In this section, we study the minimal values of $\Ch_c$ for all real $c$ no smaller than the critial threshold
\begin{align}\label{Eq:Threshold}
\thr := \rme^{D\sqrt{K^-N}/2}\, \max\!\big\lbrace\Vert\rho_0\Vert_{\Ell^\infty(\mms,\meas)}, \Vert\rho_1\Vert_{\Ell^\infty(\mms,\meas)}\big\rbrace,
\end{align}
where $D := \sup\tau(\supp\mu_0\times\supp\mu_1)$ for every $\mu_0,\mu_1\in\scrP_\comp(\mms)$ as hypothesized in \autoref{Th:Linfty estimates} (recall \autoref{Le:Bounded midpoints}). In fact, in this case it turns out that the minimal value of $\Ch_c$ is always $0$. To prove this, we first go strictly above the threshold in \autoref{Pr:Min 0}, which is where most of the work has to be done. \autoref{Cor:Thr} establishes the analogous result for the precise threshold $\thr$.

\begin{proposition}\label{Pr:Min 0} Suppose  $\smash{\wTCD_p^e(K,N)}$ for some $p\in (0,1)$, $K\in\R$ and $N\in (0,\infty)$. Let $(\mu_0,\mu_1)=(\rho_0\,\meas,\rho_1\,\meas)\in\scrP_\comp^\ac(\mms,\meas)^2$ be strongly timelike $p$-dualizable, and assume that $\rho_0,\rho_1\in\Ell^\infty(\mms,\meas)$. Finally, let $c> \thr$. Then
\begin{align*}
\min \scrE_c(\OptTGeo_{\ell_p}^\tau(\mu_0,\mu_1)) = 0.
\end{align*}
\end{proposition}

\begin{proof} We roughly follow the strategy of the proof of  \cite[Prop.~3.11]{rajala2012b}, up to several modifications required since we work entirely with timelike $\ell_p$-optimal geodesic plans and not with $\ell_p$-intermediate points. We argue by contradiction. Suppose that
\begin{align*}
\min\Ch_c(\OptTGeo_{\ell_p}^\tau(\mu_0,\mu_1))>0.
\end{align*}
Let $\smash{\Min_c \subset \OptTGeo_{\ell_p}^\tau(\mu_0,\mu_1)}$ be the set of minimizers of $\scrE_c$ on $\smash{\OptTGeo_{\ell_p}^\tau(\mu_0,\mu_1)}$, which is nonempty by \autoref{Cor:Existence minimizer}. Since all midpoints of elements of $\Min_c$ have support in the compact set $J(\mu_0,\mu_1)$, there exists $\bdpi\in \Min_c$ such that
\begin{align}\label{Eq:rhonuchoice}
\meas\big[\{\rho_\nu > c\}\big] \geq  \frac{\thr^{1/4}}{c^{1/4}}\,\sup\!\big\lbrace\meas\big[\{\rho_\omega > c\}\big]: \bdsigma\in\Min_c\big\rbrace
\end{align}
subject to the decompositions $\nu = \rho_\nu\,\meas + \nu_\perp$ and $\omega = \rho_\omega\,\meas + \omega_\perp$, employing the  ab\-bre\-via\-tions $\smash{\nu := (\eval_{1/2})_\push\bdpi}$ and $\smash{\omega := (\eval_{1/2})_\push\bdsigma}$. 

In the sequel, our strategy is to shuffle around  mass from $\bdpi$ which contributes towards the positivity of $\Ch_c(\bdpi)$ to build a timelike $\ell_p$-optimal geodesic plan from $\mu_0$ to $\mu_1$ with less energy. That way, we will arrive to a contradiction.

\textbf{Step 1.} \textit{Detection of the set of midpoints with large density.} We will first assume that $\smash{\meas\big[\{\rho_\nu > c\}\big] > 0}$, in which case the supremum on the r.h.s.~of \eqref{Eq:rhonuchoice} is strictly positive. Fix $\delta > 0$ such that
\begin{align}\label{Eq:C}
\meas\big[\{\rho_\nu > c+\delta\}] \geq \frac{\thr^{1/2}}{c^{1/2}}\,\meas\big[\{\rho_\nu > c\}\big].
\end{align}
Henceforth using the abbreviations
\begin{align*}
A &:= \{\rho_\nu > c\},\\
A_\delta &:= \{\rho_\nu > c+\delta\},\\
G_\delta &:= (\eval_{1/2})^{-1}(A_\delta),
\end{align*}
we define $\kappa_0,\kappa_1\in\scrP_\comp(\mms)$ by
\begin{align*}
\kappa_0 := \nu[A_\delta]^{-1}\,(\eval_0)_\push\big[\bdpi\mres G_\delta\big],\\
\kappa_1 := \nu[A_\delta]^{-1}\,(\eval_1)_\push\big[\bdpi\mres G_\delta\big].
\end{align*}
In other words, we take the portion of curves in $\supp\bdpi$ which hits $A_\delta$ at time $1/2$ and both trace it back to $\supp\mu_0$ and follow it forward to $\supp\mu_1$, up to normalization. It is straightforward to verify that $\kappa_i \ll \meas$, and that the density of $\kappa_i$ with respect to $\meas$ is $\meas$-essentially bounded, $i\in\{0,1\}$. Lastly,  $(\kappa_0,\kappa_1)$ is (strongly, by restriction)  timelike $p$-dualizable  by $\smash{\pi := \nu[A_\delta]^{-1}\,(\eval_0,\eval_1)_\push[\bdpi\mres G_\delta]\in\scrP(\mms^2)}$. Indeed, $\pi$ constitutes a coupling of $\kappa_0$ and $\kappa_1$ which is supported on $\smash{\mms_\ll^2}$ by \autoref{Re:Always exists}. As $\kappa_0$ and $\kappa_1$ are compactly supported, the claim thus follows from \cite[Rem.~2.20]{cavalletti2020}.

\textbf{Step 2.} \emph{Construction of a new geodesic.} By Step 1, \autoref{Le:Bounded midpoints} and \autoref{Le:Plans}, and as $\sup\tau(\supp\kappa_0\times\supp\kappa_1) \leq D$, there exists $\bdbeta\in\smash{\OptTGeo_{\ell_p}^\tau(\kappa_0,\kappa_1)}$ represen\-ting a timelike proper-time parametrized $\smash{\ell_p}$-geodesic from $\kappa_0$ to $\kappa_1$ such that
\begin{align}\label{Eq:B}
\meas\big[\{\rho > 0\}\big] \geq \frac{1}{\thr}
\end{align}
subject to the decomposition $(\eval_{1/2})_\push\bdbeta = \rho\,\meas$. Set
\begin{align*}
\bdalpha &:= \bdpi\mres (\TGeo^\tau(\mms)\setminus G_\delta) + \frac{c}{c+\delta}\,\bdpi\mres G_\delta + \frac{\delta}{c+\delta}\,\nu[A_\delta]\,\bdbeta.
\end{align*}
We verify that $\bdalpha\in \smash{\OptTGeo_{\ell_p}^\tau(\mu_0,\mu_1)}$. Clearly, $\bdalpha$ is supported on $\TGeo^\tau(\mms)$, and $(\eval_0,\eval_1)_\push\bdalpha$ is a chronological coupling of $\mu_0$ and $\mu_1$. We claim that the latter is in fact $\smash{\ell_p}$-optimal. To demonstrate this, we first note that
\begin{align*}
&\int_{\mms^2} \tau^p\d (\eval_0,\eval_1)_\push\bdalpha = \int_{\TGeo^\tau(\mms)} \tau^p\circ (\eval_0,\eval_1)\d\bdalpha\\
&\qquad\qquad = \int_{\TGeo^\tau(\mms)\setminus G_\delta}\tau^p\circ(\eval_0,\eval_1)\d\bdpi  + \frac{c}{c+\delta}\int_{\TGeo^\tau(\mms) \cap G_\delta}\tau^p\circ(\eval_0,\eval_1)\d\bdpi\\
&\qquad\qquad\qquad\qquad + \frac{\delta}{c+\delta}\,\nu[A_\delta]\int_{\TGeo^\tau(\mms)} \tau^p\circ(\eval_0,\eval_1)\d\bdbeta\\
&\qquad\qquad = \int_{\TGeo^\tau(\mms)} \tau^p\circ(\eval_0,\eval_1)\d\bdpi  - \frac{\delta}{c+\delta}\int_{\TGeo^\tau(\mms)\cap G_\delta} \tau^p\circ(\eval_0,\eval_1)\d\bdpi\\
&\qquad\qquad\qquad\qquad  + \frac{\delta}{c+\delta}\,\nu[A_\delta]\int_{\TGeo^\tau(\mms)} \tau^p\circ(\eval_0,\eval_1)\d\bdbeta\\
&\qquad\qquad= \ell_p^p(\mu_0,\mu_1) - \frac{\delta}{c+\delta}\int_{\mms^2}\tau^p \d(\eval_0,\eval_1)_\push\big[\bdpi\mres G_\delta\big]\\
&\qquad\qquad\qquad\qquad + \frac{\delta}{c+\delta}\,\nu[A_\delta]\,\ell_p^p(\kappa_0,\kappa_1).
\end{align*}
In the last step, we used that $\smash{\bdpi\in\OptTGeo_{\ell_p}^\tau(\mu_0,\mu_1)}$ and $\smash{\bdbeta\in\OptTGeo_{\ell_p}^\tau(\kappa_0,\kappa_1)}$. Now note that $\smash{\bdpi \mres G_\delta\leq \bdpi}$, and given that $\smash{\bdpi\in\OptTGeo_{\ell_p}^\tau(\mu_0,\mu_1)}$, by  \autoref{Le:Plans} $\nu[A_\delta]^{-1}\,\bdpi\mres G_\delta$ constitutes a timelike $\smash{\ell_p}$-optimal geodesic plan interpolating its mar\-ginals. The latter are precisely $\kappa_0$ and $\kappa_1$, whence
\begin{align*}
\int_{\mms^2} \tau^p\d(\eval_0,\eval_1)_\push \big[\bdpi\mres G_\delta\big] = \nu[A_\delta]\,\ell_p^p(\kappa_0,\kappa_1),
\end{align*}
which proves the $\smash{\ell_p}$-optimality of $(\eval_0,\eval_1)_\push\bdalpha$.

\textbf{Step 3.} \textit{Energy excess of $\bdalpha$.} We decompose $\theta = \rho_\theta\,\meas + \theta_\perp$, where $\theta := (\eval_{1/2})_\push\bdalpha$, and compute
\begin{align*}
&\scrE_c(\bdpi) - \scrE_c(\bdalpha) = \int_\mms (\rho_\nu-c)^+\d\meas + \nu_\perp[\mms] - \int_\mms (\rho_\theta-c)^+\d\meas - \theta_\perp[\mms]\\
&\qquad\qquad =\int_{\mms\setminus A_\delta} \Big[(\rho_\nu -c)^+ - \Big[\rho_\nu + \frac{\delta}{c+\delta}\,\nu[A_\delta]\,\rho - c\Big]^+\Big]\d\meas\\
&\qquad\qquad\qquad\qquad + \int_{A_\delta} \Big[(\rho_\nu-c)^+ - \Big[\frac{c}{c+\delta}\,\rho_\nu + \frac{\delta}{c+\delta}\,\nu[A_\delta]\,\rho -c\Big]^+\Big]\d\meas\\
&\qquad\qquad = \int_{\mms\setminus A_\delta} \Big[(\rho_\nu -c)^+ - \Big[\rho_\nu + \frac{\delta}{c+\delta}\,\nu[A_\delta]\,\rho - c\Big]^+\Big]\d\meas\\
&\qquad\qquad\qquad\qquad + \frac{\delta}{c+\delta}\int_{A_\delta} \big[\rho_\nu - \nu[A_\delta]\,\rho\big]\d\meas\\
&\qquad\qquad = \int_{\mms\setminus A_\delta} \Big[(\rho_\nu-c)^+ - \Big[\rho_\nu + \frac{\delta}{c+\delta}\,\nu[A_\delta]\,\rho - c\Big]^+\Big]\d\meas\\
&\qquad\qquad\qquad\qquad  + \frac{\delta}{c+\delta}\,\nu[A_\delta]\int_{\mms\setminus A_\delta} \rho \d\meas\\
&\qquad\qquad= \int_{B_1} (c-\rho)\d\meas + \frac{\delta}{c+\delta}\,\nu[A_\delta]\int_{B_2}\rho\d\meas\\
&\qquad\qquad= \int_{\{\rho < c\}} \min\!\Big\lbrace c-\rho, \frac{\delta}{c+\delta}\,\nu[A_\delta]\,\rho\Big\rbrace\d\meas,
\end{align*}
where we abbreviate
\begin{align*}
B_1 &:= \Big\lbrace \rho_\nu < c \leq \rho_\nu + \frac{\delta}{c+\delta}\,\nu[A_\delta]\,\rho\Big\rbrace,\\
B_2 &:= \Big\lbrace \rho_\nu + \frac{\delta}{c+\delta}\,\nu[A_\delta]\,\rho < c\Big\rbrace.
\end{align*}
Now we consider the set
\begin{align*}
E := \{\rho >0\}.
\end{align*}
Since $\bdpi$ minimizes $\scrE_c$, the last integral must vanish identically, whence
\begin{align}\label{Eq:A}
\meas\big[\{\rho < c\}\cap E\big] =0
\end{align}
and $\bdalpha\in\Min_c$. On the other hand, $\rho_\theta > c$ on  $\{\rho \geq c\} \cap E$. Combining \eqref{Eq:A}, \eqref{Eq:B}, \eqref{Eq:C} and \eqref{Eq:rhonuchoice} therefore gives
\begin{align*}
\meas\big[\{\rho_\theta > c\}\big] &\geq \meas\big[\{\rho \geq c\} \cap E\big] = \meas[E] \geq \frac{\nu[A_\delta]}{\thr} \geq \frac{c}{\thr}\,\meas[A_\delta]\\
&\geq \frac{c^{1/2}}{\thr^{1/2}}\,\meas[A] \geq \frac{c^{1/4}}{\thr^{1/4}}\,\sup\!\big\lbrace \meas\big[\{\rho_\omega > c\}\big] : \bdsigma\in\Min_\rmc\big\rbrace,
\end{align*}
which yields the desired contradiction.

\textbf{Step 4.} \textit{Treatise of the singular part.} In the remaining case $\smash{\meas\big[\{\rho_\nu > c\}\big] = 0}$,  the measure $\nu$ has a nontrivial singular part with respect to $\meas$ since $\scrE_c(\bdpi)>0$. Analogously to above, we can shuffle this singular portion to the $\bdbeta$-part of the timelike proper-time parametrized $\smash{\ell_p}$-geodesic $\bdalpha$ constructed in Step 2 using \eqref{Eq:B} and giving $\Ch_c(\bdalpha) < \Ch_c(\bdpi)$, which leads to a contradiction.
\end{proof}

\begin{corollary}\label{Cor:Thr} Under the same assumptions as in \autoref{Pr:Min 0},
\begin{align*}
\min\Ch_{\thr}(\OptTGeo_{\ell_p}^\tau(\mu_0,\mu_1))=0.
\end{align*}
\end{corollary}

\begin{proof} By \autoref{Pr:Min 0}, we get that that for every $n\in\N$,
\begin{align*}
&\min\Ch_{\thr}(\OptTGeo_{\ell_p}^\tau(\mu_0,\mu_1))\\
&\qquad\qquad \leq \min\Ch_{\thr + 2^{-n}}(\OptTGeo_{\ell_p}^\tau(\mu_0,\mu_1)) + 2^{-n}\,\meas\big[J(\mu_0,\mu_1)\big]\\
&\qquad\qquad = 2^{-n}\,\meas\big[J(\mu_0,\mu_1)\big].\textcolor{white}{\Int_{\ell_p}^{1/2}}
\end{align*}
The r.h.s.~converges to zero as $n\to\infty$ by compactness of $J(\mu_0,\mu_1)$.
\end{proof}

\subsubsection{Maximizers of $\scrV_N$ have zero excess} For the subsequent main result of this section, recall the definitions \eqref{Eq:scrV} of $\scrV_N$ and \eqref{Eq:Threshold} of $\thr$, respectively.

\begin{proposition}\label{Pr:Max = 0} Suppose  $\smash{\wTCD_p^e(K,N)}$ for some $p\in (0,1)$, $K\in\R$, and $N\in (0,\infty)$. Let $(\mu_0,\mu_1)=(\rho_0\,\meas,\rho_1\,\meas)\in\scrP_\comp^\ac(\mms,\meas)^2$ be strongly timelike $p$-dua\-li\-za\-ble, and assume that $\rho_0,\rho_1\in\Ell^\infty(\mms,\meas)$.  Then 
\begin{align*}
\scrE_\thr(\bdpi)=0
\end{align*}
for every maximizer $\smash{\bdpi\in\OptTGeo_{\ell_p}^\tau(\mu_0,\mu_1)}$ of $\scrV_N$. 
\end{proposition}

\begin{proof} As for \autoref{Pr:Min 0}, our ansatz is a  contradiction argument, i.e.~we assume the existence of a maximizer $\smash{\bdpi\in\OptTGeo_{\ell_p}^\tau(\mu_0,\mu_1)}$ with
\begin{align*}
\scrE_\thr(\bdpi)> 0.
\end{align*}
Our proof follows the one of \cite[Prop.~3.5]{rajala2012a} and for \autoref{Pr:Min 0}.

\textbf{Step 1.} \textit{Detection of the set of midpoints with large density.} We recall from \autoref{Le:Maximizer} that by the hypothesized $\wTCD$ condition and since all timelike proper-time parametrized $\smash{\ell_p}$-geodesics from $\mu_0$ to $\mu_1$ have support in the compact set $J(\mu_0,\mu_1)$, we must have $\scrV_N(\bdpi) > 0$. In particular, $\nu := (\eval_{1/2})_\push\bdpi \ll \meas$. Now let $\eta > 0$ such that
\begin{align*}
\meas\big[\{\rho_\nu > \thr+\eta\}\big] \geq \meas\big[\{\rho_\nu > \thr + 2\eta\}\big] > 0
\end{align*}
subject to the decomposition $\nu = \rho_\nu\,\meas$, and define
\begin{align*}
c_1 := \frac{4}{\eta}\, \meas\big[\{\rho_\nu > \thr+\eta\}\big] - \frac{4}{\eta}\, \meas\big[\{\rho_\nu > \thr + 2\eta\}\big].
\end{align*}
Given any $\phi \in (0,\eta/3)$ there exists $\delta\in (\eta,2\eta)$ such that
\begin{align*}
\meas[A_\delta] < \meas[A_\delta'] + c_1\phi,
\end{align*}
where we have defined
\begin{align*}
A_\delta &:= \{\rho_\nu > \thr+\delta\},\\
A_\delta' &:= \{\rho_\nu > \thr + \delta - 3\phi\}.
\end{align*}

\textbf{Step 2.} \textit{Construction of a new geodesic.} Let $\kappa_0,\kappa_1\in\scrP_\comp(\mms)$ be defined as in Step 1 of the proof of \autoref{Pr:Min 0}. Using \autoref{Cor:Thr}, there exists $\smash{\bdbeta\in\OptTGeo_{\ell_p}^\tau(\mu_0,\mu_1)}$ such that 
\begin{align*}
\Vert\rho\Vert_{\Ell^\infty(\mms,\meas)}\leq \frac{\thr}{\nu[A_\delta]} 
\end{align*}
subject to the decomposition $(\eval_{1/2})_\push\bdbeta = \rho\,\meas$. Setting $\smash{G_\delta := (\eval_{1/2})^{-1}(A_\delta)}$, define  $\smash{\bdalpha\in\OptTGeo_{\ell_p}^\tau(\mu_0,\mu_1)}$ through
\begin{align*}
\bdalpha := \bdpi\mres(\TGeo^\tau(\mms)\setminus G_\delta) + \frac{\thr+\delta-\phi}{\thr+\delta}\,\bdpi\mres G_\delta + \frac{\phi}{\thr+\delta}\,\nu[A_\delta]\,\bdbeta.
\end{align*}

\textbf{Step 3.} \textit{Energy excess of $\bdalpha$.} We decompose $\theta = \rho_\theta\,\meas + \theta_\perp$, where $\theta := (\eval_{1/2})_\push\bdalpha$, and let $\rho$ denote the density of $(\eval_{1/2})_\push\bdbeta$ with respect to $\meas$. Then the subsequent estimates are readily verified.
\begin{itemize}
\item On $A_\delta$, we have
\begin{align*}
\rho_\theta &\leq \frac{\thr+\delta - \phi}{\thr+\delta}\,\rho_\nu + \frac{\phi}{\thr+\delta}\,\nu[A_\delta]\,\rho_\theta \leq \frac{(\thr+\delta-\phi)\rho_\nu + \thr\,\phi}{\thr+\delta} \\
&\leq \rho_\nu + \frac{(\thr-\rho_\nu)\phi}{\thr+\delta} < \rho - \frac{\delta\phi}{\thr+\delta},\\
\rho_\theta &\geq \frac{\thr+\delta-\phi}{\thr+\delta}\,\rho_\nu > \thr+\delta-\phi.
\end{align*}
\item On $A_\delta'\setminus A_\delta$, we have
\begin{align*}
\rho_\theta \leq \rho_\nu + \frac{\phi}{\thr+\delta}\,\nu[A_\delta]\,\rho \leq \rho_\nu + \frac{\thr\,\phi}{\thr+\delta} < \thr+\delta+\phi.
\end{align*}
\item On $\mms\setminus A_\delta'$, we have
\begin{align*}
\rho_\theta & \leq \rho_\nu + \frac{\phi}{\thr+\delta}\,\nu[A_\delta]\,\rho \leq \thr + \delta-3\phi + \frac{\thr\,\phi}{\thr+\delta}\leq\thr+\delta-2\phi.
\end{align*}
\end{itemize}
Moreover, we define and estimate the mass differences
\begin{align*}
\kappa_{A_\delta} &:= \int_{A_\delta} (\rho_\nu-\rho_\theta)\d\meas \geq c_2\,\phi,\\
\kappa_{A_\delta'\setminus A_\delta} &:= \int_{A_\delta'\setminus A_\delta} (\rho_\theta - \rho_\nu)\d\meas < c_1\,\phi^2,\\
\kappa_{\mms\setminus A_\delta'} &:= \int_{\mms\setminus A_\delta'} (\rho_\theta -\rho_\nu)\d\meas,
\end{align*}
where $c_2 := \delta\,\meas[A_\delta]/(\thr+\delta)$. Approximating $\Ent_\meas$ by a Rényi-type entropy \cite[Lem.~4.1]{sturm2006a} and with analogous computations as for \cite[Prop.~3.5]{rajala2012a}, for every $\varepsilon > 0$ there exists $N_\varepsilon \geq N$ such that
\begin{align*}
&\Ent_\meas(\theta) - \Ent_\meas(\nu) \leq \varepsilon -N_\varepsilon\int_\mms \rho_\theta^{1-1/N_\varepsilon}\d\meas + N_\varepsilon\int_\mms \rho_\nu^{1-1/N_\varepsilon}\d\meas\\
&\qquad\qquad\leq \varepsilon +N_\varepsilon \int_{A_\delta} \rho_\theta^{-1/N_\varepsilon}\,(\rho_\nu-\rho_\theta)\d\meas + N_\varepsilon\int_{\mms\setminus A_\delta'} \rho_\theta^{-1/N_\varepsilon}\,(\rho_\nu-\rho_\theta)\d\meas\\
&\qquad\qquad\qquad\qquad + N_\varepsilon\int_{A_\delta'\setminus A_\delta} \rho_\theta^{-1/N_\varepsilon}\,(\rho_\nu-\rho_\theta)\d\meas\\
&\qquad\qquad\leq \varepsilon+ N_\varepsilon\,\kappa_{A_\delta}\,(\thr+\delta-\phi)^{-1/N_\varepsilon} - N_\varepsilon\,\kappa_{\mms\setminus A_\delta'}\,(\thr+\delta-2\phi)^{-1/N_\varepsilon}\\
&\qquad\qquad\qquad\qquad - N_\varepsilon\, \kappa_{A_\delta'\setminus A_\delta}\,(\thr+\delta+\phi)^{-1/N_\varepsilon}\textcolor{white}{\int_{A_\delta}^A}\\
&\qquad\qquad = N_\varepsilon\,\kappa_{A_\delta}\,\big[(\thr+\delta-\phi)^{-1/N_\varepsilon} - (\thr+\delta-2\phi)^{-1/N_\varepsilon}\big]\\
&\qquad\qquad\qquad\qquad + N_\varepsilon\,\kappa_{A_\delta'\setminus A_\delta} \,\big[(\thr+\delta-2\phi)^{-1/N_\varepsilon} - (\thr+\delta+\phi)^{-1/N_\varepsilon}\big]\!\!\!\textcolor{white}{\int^A}\\
&\qquad\qquad\leq \varepsilon - \frac{c_2 - c_1\phi}{(\thr+\delta-2\phi)}\,\phi^2 + c_2c_3\,\phi^3.
\end{align*}
Here $c_3 > 0$ is some constant independent of $\varepsilon$, and the last inequality is computed as in the proof of \cite[Prop.~3.5]{rajala2012a}. Choosing $\varepsilon$ and $\phi$ small enough, the r.h.s.~becomes strictly negative. Therefore $\scrV_N(\bdalpha) > \scrV_N(\bdpi)$, 
which is the desired contradiction.
\end{proof}

The following consequence thus terminates the \textit{proof} of \autoref{Th:Linfty estimates}.

\begin{corollary}\label{Cor:Cor} Retain the assumptions and the notation from \autoref{Pr:Max = 0}. Then the candidate $(\mu_t)_{t\in[0,1]}$ constructed in \autoref{Sub:Construction} satisfies, for every $t\in[0,1]$,  $\mu_t=\rho_t\,\meas\in\Dom(\Ent_\meas)$ as well as
\begin{align*}
\Vert\rho_t\Vert_{\Ell^\infty(\mms,\meas)} \leq \rme^{D\sqrt{K^-N}/2}\,\max\!\big\lbrace \Vert\rho_0\Vert_{\Ell^\infty(\mms,\meas)}, \Vert\rho_1\Vert_{\Ell^\infty(\mms,\meas)}\big\rbrace.
\end{align*}
\end{corollary}

\begin{remark}\label{Re:Minor} Minor modifications of the above arguments give a $\TCD$ version of \autoref{Th:Linfty estimates}, namely assuming $\smash{\TCD_p^e(K,N)}$ instead of $\smash{\wTCD_p^e(K,N)}$ therein, and that --- instead of being strongly timelike $p$-dualizable --- every $\smash{\ell_p}$-optimal coupling of $\mu_0$ and $\mu_1$ is chronological.

To see this, first note that item (v) of \autoref{Le:Plans} merely needs all $\smash{\ell_p}$-optimal couplings of $\mu_0$ and $\mu_1$ to be chronological \cite[Prop.~B.11]{braun+}. Second, similarly as in \autoref{Le:Strong}, the property of pairs admitting only chronological $\smash{\ell_p}$-optimal couplings propagates through proper-time parametrized $\smash{\ell_p}$-geodesics $(\mu_t)_{t\in[0,1]}$. Indeed, if  one $\smash{\ell_p}$-optimal coupling of $(\mu_s,\mu_t)$ is not chronological for some $s,t\in[0,1]$ with $s<t$, restricting it to null related point pairs and using a gluing procedure we could produce a measure $\bdpi$ on $\scrP(\Cont([0,1];\mms))$ concentrated on maximizing causal curves such that $(\eval_0,\eval_1)_\push\bdpi$ is $\smash{\ell_p}$-optimal. But the latter must be chronological by assumption, and since $\bdpi$-a.e.~curve changes its character from timelike to null and back to timelike, this contradicts regularity of $(\mms,\met,\ll,\leq,\tau)$.
\end{remark}

\begin{remark} Similar arguments as above give the existence of good geodesics for arbitrary metric measure spaces obeying $\smash{\CD^e(K,N)}$. While in the essentially nonbranching case, this partly follows from \cite[Thm.~1.2]{rajala2012b} by \cite[Thm.~3.12]{erbar2015}, we are not aware of such a general result for the \emph{entropic} $\CD$  condition.
\end{remark}

\section{Variations of the main result}\label{Ch:Variations}

Finally, we discuss various extensions of \autoref{Th:Linfty estimates}. In all cases, the proof of \autoref{Th:Linfty estimates} can then mostly be adapted to the respective situation up to some details which we highlight below.

\subsection{The infinite-dimensional case}\label{Sub:Infinite dim} %Unlike the $(K,N,p)$-convexity inequality in \autoref{Th:Linfty estimates}, the sole existence of an $\smash{\ell_p}$-geodesic with uniformly bounded densities does not require a synthetic upper bound on the dimension of $\mms$. 
%Here we prove a corresponding version of \autoref{Th:Linfty estimates}, cf.~\autoref{Th:Ent linfty}, under the ``infinite-dimensional'' timelike curvature-dimension condition introduced in \autoref{Def:CD K infty}. 
The following is a Lorentzian  analogue of Sturm's $\CD(K,\infty)$ condition for metric measure spaces \cite[Def.~4.5]{sturm2006a} (see also \cite{lott2009}); the counterpart of \autoref{Th:Ent linfty} for metric measure spaces is due to \cite[Thm.~1.3]{rajala2012a} whose strategy we loosely follow. 

%\subsubsection{The $\smash{\wTCD_p(K,\infty)}$ condition} %For technical reasons, see \autoref{Re:Pathological}, the following  is leaned on the \emph{weak} timelike curvature-dimension condition.

\begin{definition}\label{Def:CD K infty} A measured Lorentzian pre-length space $(\mms,\met,\meas,\ll,\leq,\tau)$ is termed to obey the \emph{weak timelike curvature-dimension condition} $\smash{\wTCD_p(K,\infty)}$ for $p\in (0,1)$ and $K\in\R$ if for every strongly timelike $p$-dualizable pair $(\mu_0,\mu_1)\in(\scrP_\comp(\mms)\cap\Dom(\Ent_\meas))^2$, there exists a timelike $p$-dualizing coupling $\smash{\pi\in\Pi_\ll(\mu_0,\mu_1)}$ and a timelike proper-time parametrized $\smash{\ell_p}$-geodesic $(\mu_t)_{t\in [0,1]}$  such that, for every $t\in[0,1]$,
\begin{align*}
\Ent_\meas(\mu_t) \leq (1-t)\Ent_\meas(\mu_0) + t\Ent_\meas(\mu_1) - \frac{K}{2}\,t(1-t)\,T_\pi^2.
\end{align*}
\end{definition}

\begin{remark}\label{REEEE} In an evident way, one may define the $\smash{\TCD_p(K,\infty)}$ condition as an infinite-dimensional analogue of $\smash{\TCD_p^e(K,N)}$. Taking \autoref{Re:Minor} into account, similar results as those presented below for $\smash{\wTCD_p(K,\infty)}$ hold for this curvature-dimension condition as well, up to minor modifications.
\end{remark}

Before turning to the main   \autoref{Th:Ent linfty} of this section, independently of it we examine  elementary properties of the $\smash{\wTCD_p(K,\infty)}$ condition just introduced. The reader may  directly go over to \autoref{Sub:MRES} at first reading.

\subsubsection{From finite to infinite dimension} It is clear that $\wTCD_p(K,\infty)$ has analogous consistency and scaling properties as its finite-dimensional Lorentzian  counterpart \cite[Lem.~3.10]{cavalletti2020}. Moreover, as already indicated in \autoref{REEEE} it can be regarded as an ``infinite-dimensional'' analogue of the $\wTCD$ condition from \autoref{Def:TCD} by the following result.

\begin{proposition} The condition  $\smash{\wTCD_p^e(K,N)}$ implies $\smash{\wTCD_p(K,\infty)}$ for every $p\in (0,1)$ and every $K\in\R$.
\end{proposition}

\begin{proof} We follow the argument for \cite[Lem.~2.12]{erbar2015}. By nestedness of the weak timelike curvature-dimension condition \cite[Lem.~3.10]{cavalletti2020}, given any strongly timelike $p$-dualizable pair $\smash{(\mu_0,\mu_1)\in(\scrP_\comp(\mms)\cap\Dom(\Ent_\meas))^2}$, there exists a   timelike $p$-dualizing $\pi\in\Pi_\ll(\mu_0,\mu_1)$ and a $\Dom(\Ent_\meas)$-valued timelike proper-time parametrized $\smash{\ell_p}$-geodesic $(\mu_t)_{t\in [0,1]}$ from $\mu_0$ to $\mu_1$ such that, for every $N'\geq N$, we have
\begin{align}\label{Eq:Subtract}
\scrU_{N'}(\mu_t) \geq \sigma_{K,N'}^{(1-t)}(T_\pi)\,\scrU_{N'}(\mu_0) + \sigma_{K,N'}^{(t)}(T_\pi)\,\scrU_{N'}(\mu_1).
\end{align}
Note that $\pi$ and $\smash{(\mu_t)_{t\in[0,1]}}$ can be chosen independently of $N'$. Using that
\begin{align*}
\sigma_{K,N'}^{(t)}(\vartheta) &= t -\frac{K}{6N'}\,(t^3-t)\,\vartheta^2 + \rmo((N')^{-1}),\\
\scrU_{N'}(\mu) &= 1-\frac{1}{N'}\Ent_\meas(\mu) + \rmo((N')^{-1})
\end{align*}
for every $\mu\in\Dom(\Ent_\meas)$ as $N'\to\infty$, the claim follows by subtracting $1$ at both sides of \eqref{Eq:Subtract}, multiplying the resulting inequality by $N'$, and letting $N' \to\infty$.
\end{proof}

\subsubsection{Geodesics with uniformly bounded densities}\label{Sub:MRES} Now we turn to our actual goal, namely \autoref{Th:Ent linfty}. Its statement holds in a stronger form, cf. \autoref{Re:BL}, but we prefer to present a slightly different proof for the sole existence of timelike proper-time parametrized $\smash{\ell_p}$-geodesics with bounded densities. This underlines the role that $\smash{\ell_p}$-geodesics play also in the $\smash{\wTCD_p(K,\infty)}$ case and is exemplary for a similar result in the next \autoref{Sub:Conv} where, however, the relevant exponential term does not appear.

Reproducing the proof of \autoref{Le:Bounded midpoints} gives the following.

\begin{lemma}\label{Le:Lemma} Let $(\mms,\met,\meas,\ll,\leq,\tau)$ obey $\smash{\wTCD_p(K,\infty)}$ for some $p\in (0,1)$ and $K\in\R$. Assume that $\smash{(\mu_0,\mu_1)=(\rho_0\,\meas,\rho_1\,\meas)\in\scrP_\comp^\ac(\mms,\meas)^2}$ is strongly timelike $p$-dualizable with $\rho_0,\rho_1\in\Ell^\infty(\mms,\meas)$. Then there is a timelike proper-time parametrized $\smash{\ell_p}$-geodesic $(\mu_t)_{t\in [0,1]}$ from $\mu_0$ to $\mu_1$ such that $\mu_t=\rho_t\,\meas\in\Dom(\Ent_\meas)$ for every $t\in (0,1)$, and
\begin{align*}
\meas\big[\{\rho_{1/2}>0\}\big] \geq \rme^{-K^-D^2/8}\,\max\!\big\lbrace\Vert\rho_0\Vert_{\Ell^\infty(\mms,\meas)}, \Vert \rho_1\Vert_{\Ell^\infty(\mms,\meas)}\big\rbrace^{-1},
\end{align*}
where $D\geq \sup\tau(\supp\mu_0\times\supp\mu_1)$.
\end{lemma}

%Given this \autoref{Le:Lemma}, we only outline the proof of the next \autoref{Th:Ent linfty} and highlight the necessary changes compared to our arguments in  \autoref{Ch:Good}.

\begin{theorem}\label{Th:Ent linfty} Suppose $\smash{\wTCD_p(K,\infty)}$ for $p\in (0,1)$ and $K\in\R$. Let $(\mu_0,\mu_1)=(\rho_0\,\meas,\rho_1\,\meas)\in\smash{\scrP_\comp^\ac(\mms,\meas)^2}$ be strongly timelike $p$-dualizable, and assume that $\rho_0,\rho_1\in\Ell^\infty(\mms,\meas)$.  Then there is a timelike proper-time parametrized $\smash{\ell_p}$-geodesic $(\mu_t)_{t\in [0,1]}$ connecting $\mu_0$ and $\mu_1$ such that for every $t\in [0,1]$, $\mu_t\in\Dom(\Ent_\meas)$ and
\begin{align}\label{Eq:BSDAJ}
\Vert\rho_t\Vert_{\Ell^\infty(\mms,\meas)} \leq \rme^{K^-D^2/12}\,\max\!\big\lbrace\Vert\rho_0\Vert_{\Ell^\infty(\mms,\meas)},\Vert\rho_1\Vert_{\Ell^\infty(\mms,\meas)}\big\rbrace,
\end{align}
where $D:= \sup\tau(\supp\mu_0\times\supp\mu_1)$.
\end{theorem}

\begin{proof} We only outline the proof and highlight the necessary changes compared to our arguments in  \autoref{Ch:Good}. Let us redefine
\begin{align*}
\thr :=  \rme^{K^-D^2/8}\,\max\!\big\lbrace \Vert\rho_0\Vert_{\Ell^\infty(\mms,\meas)}, \Vert\rho_1\Vert_{\Ell^\infty(\mms,\meas)}\big\rbrace.
\end{align*}
Unlike \autoref{Sub:Construction}, here we directly construct a candidate by selecting our midpoints as minimizers of the functional $\scrE_\thr$ from \eqref{Eq:Functional E}. Let $\smash{\bdpi\in\OptTGeo_{\ell_p}^\tau(\mu_0,\mu_1)}$ be a minimizer of $\scrE_\thr$ according to \autoref{Cor:Existence minimizer}. Observe that  \autoref{Cor:Thr} still holds in this case, where the modified threshold $\thr$ comes from \autoref{Le:Lemma}. Define  $\mu_{1/2} := (\eval_{1/2})_\push\bdpi\in\scrP_\comp(\mms)$. By \autoref{Cor:Thr}, we have $\mu_{1/2}\in\Dom(\Ent_\meas)$  and
\begin{align*}
\Vert \rho_{1/2}\Vert_{\Ell^\infty(\mms,\meas)} \leq \rme^{K^-D^2/8}\,\max\!\big\lbrace \Vert \rho_0\Vert_{\Ell^\infty(\mms,\meas)},\Vert \rho_1\Vert_{\Ell^\infty(\mms,\meas)}\big\rbrace
\end{align*}
subject to the decomposition $\mu_{1/2}=\rho_{1/2}\,\meas$. By \autoref{Le:Strong}, the pairs $(\mu_0,\mu_{1/2})$ and $(\mu_{1/2},\mu_1)$ are strongly timelike $p$-dualizable. The con\-struction of $\mu_{1/2}$ yields that $\smash{\sup \tau(\supp\mu_0\times\supp\mu_{1/2})}$ and $\smash{\sup\tau(\supp\mu_{1/2}\times\supp\mu_1)}$ are no larger than $D/2$. Moreover, $\mu_{1/2}$ is an $1/2$-midpoint with respect to $\smash{\ell_p}$. Next, we  construct $\mu_{1/4}\in\scrP_\comp(\mms)$ and $\mu_{3/4}\in\scrP_\comp(\mms)$ as above as midpoints of some element of $\smash{\OptTGeo_{\ell_p}^\tau(\mu_0,\mu_{1/2})}$ and $\smash{\OptTGeo_{\ell_p}^\tau(\mu_{1/2},\mu_1)}$ according to \autoref{Cor:Existence minimizer}, \autoref{Cor:Thr}, and \autoref{Le:Strong}, respectively. Proceeding iteratively in this way after gluing, as in \autoref{Sub:Construction}, we get a timelike proper-time parametrized $\smash{\ell_p}$-geodesic $(\mu_t)_{t\in [0,1]}$ with the following properties. For every $t\in [0,1]\cap\boldsymbol{\mathrm{D}}$ written as $\smash{t = k\,2^{-n}}$, $n\in\N$ and odd $k\in\{1,\dots,2^n-1\}$, we have $\smash{\mu_t\in\Dom(\Ent_\meas)}$ with
\begin{align}\label{Eq:RHOT}
\begin{split}
\Vert \rho_t\Vert_{\Ell^\infty(\mms,\meas)} &\leq  \rme^{4^{-n+1}K^-D/8}\,\max\!\big\lbrace\Vert \rho_{(k-1)2^{-n}}\Vert_{\Ell^\infty(\mms,\meas)},\\
&\qquad\qquad \Vert\rho_{(k+1)2^{-n}}\Vert_{\Ell^\infty(\mms,\meas)}\big\rbrace
\end{split}
%\Vert \rho_t\Vert_{\Ell^\infty(\mms,\meas)} &\leq \Big[\!\prod_{i=1}^n \rme^{2^{-i-2}D\sqrt{K^-N}}\Big]\max\!\big\lbrace\Vert \rho_0\Vert_{\Ell^\infty(\mms,\meas)}, \Vert\rho_1\Vert_{\Ell^\infty(\mms,\meas)}\big\rbrace
\end{align}
subject to the decomposition $\mu_s = \rho_s\,\meas$ for all $s\in[0,1]$ under consideration. Here we have used that by our midpoint construction along timelike $\smash{\ell_p}$-optimal geodesic plans, for every $n\in\N$ and every odd $k\in\{0,\dots,2^n\}$ the function $\tau$ is no larger than $\smash{2^{-n+1}\,D}$ on $\smash{\supp \mu_{(k-1)2^{-n}} \times \supp\mu_{(k+1)2^{-n}}}$. Inductively, \eqref{Eq:BSDAJ} holds for every $t\in [0,1]\cap \boldsymbol{\mathrm{D}}$.

By weak lower semicontinuity of the functional $\scrF_\thr$ from \eqref{Eq:F Functional} on $\scrP(J(\mu_0,\mu_1))$, see the proof of \autoref{Cor:Existence minimizer}, and \eqref{Eq:RHOT} we  get $\scrF_\thr(\mu_t)=0$ for every $t\in[0,1]$. This  implies $\mu_t\in\Dom(\Ent_\meas)$ and \eqref{Eq:RHOT} for its density with respect to $\meas$.
\end{proof}

\begin{remark}\label{Re:BL} Combining the arguments of \cite[Ch.~4]{ambrosiomondino2015} with our strategy in \autoref{Ch:Good}, one can construct timelike proper-time parametrized $\smash{\ell_p}$-geodesics satisfying the conclusion of \autoref{Th:Ent linfty} along which, in addition, the semiconvexity inequality for $\Ent_\meas$ defining $\smash{\wTCD_p(K,\infty)}$ holds.
\end{remark}

\subsection{General timelike convex functionals}\label{Sub:Conv} %\comment{Also works for $(0,N)$-convex functionals, probably the definitions below are thus unneeded}
Similar conclusions as in \autoref{Th:Ent linfty} can also be drawn for any kind of functional obeying a certain timelike convexity property as follows. Let $f \colon [0,\infty)\to\R$ be convex with $f(0)=0$, and define the functional $\rmE\colon \scrP(\mms)\to [-\infty,\infty]$ by
\begin{align*}
\rmE(\mu) := \begin{cases} \displaystyle\int_\mms f(\rho)\d\meas & \textnormal{if } \mu = \rho\,\meas \ll \meas,\  f^+(\rho)\in\Ell^1(\mms,\meas),\\
\infty & \textnormal{otherwise}.
\end{cases}
\end{align*}

\begin{theorem}\label{Th:CONV} Let $\rmE$ be weakly $(0,N,p)$-convex relative to $\scrP_\comp(\mms)^2\cap\STD_p(\mms)$ according to \autoref{Def:Convex}, $p\in (0,1)$ and $N\in (0,\infty]$, and assume that $r\mapsto f(r)/r$ is strictly increasing on $(0,\infty)$.  Then for every strongly timelike $p$-dualizable pair $(\mu_0,\mu_1)=(\rho_0\,\meas,\rho_1\,\meas)\in \scrP_\comp^\ac(\mms,\meas)$ with $\rho_0,\rho_1\in\Ell^\infty(\mms,\meas)$, there exists a timelike proper-time parametrized $\smash{\ell_p}$-geodesic $(\mu_t)_{t\in [0,1]}$ such that for every $t\in (0,1)$, we have $\mu_t=\rho_t\,\meas\in \Dom(\rmE)$ and
\begin{align*}
\Vert \rho_t\Vert_{\Ell^\infty(\mms,\meas)} \leq \max\!\big\lbrace \Vert\rho_0\Vert_{\Ell^\infty(\mms,\meas)}, \Vert \rho_1\Vert_{\Ell^\infty(\mms,\meas)}\big\rbrace.
\end{align*}
\end{theorem}

Again, the \emph{proof} of \autoref{Th:CONV} is mainly based on the following result. The rest is treated analogously to \autoref{Ch:Good} and \autoref{Sub:Infinite dim}.

\begin{lemma} Let $\rmE$ be as hypothesized in \autoref{Th:CONV}. Let $(\mu_0,\mu_1)=(\rho_0\,\meas,\rho_1\,\meas)\in\scrP_\comp^\ac(\mms,\meas)^2$ be strongly timelike $p$-dualizable with $\rho_0,\rho_1\in\Ell^\infty(\mms,\meas)$. Then there is a timelike proper-time parametrized $\ell_p$-geodesic $(\mu_t)_{t\in [0,1]}$ connecting $\mu_0$ and $\mu_1$ such that for every $t\in (0,1)$, $\mu_t=\rho_t\,\meas\in\Dom(\rmE)$ and
\begin{align*}
\meas\big[\{\rho_t>0\}\big]  \geq \max\!\big\lbrace\Vert\rho_0\Vert_{\Ell^\infty(\mms,\meas)},\Vert \rho_1\Vert_{\Ell^\infty(\mms,\meas)}\big\rbrace^{-1}.
\end{align*}
\end{lemma}

\begin{proof} By our assumptions on $\rho_0$, $\rho_1$, and  $f$ it is clear that $\mu_0,\mu_1\in\Dom(\rmE)$.

Let $\smash{\pi\in\Pi_\ll(\mu_0,\mu_1)}$ be a timelike $p$-dualizing coupling and $(\mu_t)_{t\in [0,1]}$ be a timelike proper-time parametrized $\smash{\ell_p}$-geodesic with respect to which $\rmE$ satisfies its defining convexity inequality. As in the proof of \autoref{Le:Bounded midpoints}, we infer that $\mu_t\in\Dom(\rmE)$ for every $t\in (0,1)$, and that 
\begin{align*}
E := \{\rho_t >0\},
\end{align*}
subject to the decomposition $\mu_t=\rho_t\,\meas$, obeys $\meas[E]\in (0,\infty)$. Setting
\begin{align*}
R := \max\!\big\lbrace \Vert\rho_0\Vert_{\Ell^\infty(\mms,\meas)},\Vert\rho_1\Vert_{\Ell^\infty(\mms,\meas)}\big\rbrace
\end{align*}
then yields, on the one hand,
\begin{align*}
\rmE(\mu_t) \leq  (1-t)\int_{\{\rho_0>0\}} \frac{f(\rho_0)}{\rho_0}\,\rho_0\d\meas + t\int_{\{\rho_1 >0\}} \frac{f(\rho_1)}{\rho_1}\,\rho_1\d\meas \leq \frac{f(R)}{R}.
\end{align*}
On the other hand, $\rmE(\mu_t) \geq \meas[E]\,f(\meas[E]^{-1})$ by Jensen's inequality. Employing the strict monotonicity of $\smash{r\mapsto f(r)/r}$ on $(0,\infty)$ yields the claim.
\end{proof}

\subsection{Timelike measure-contraction property}\label{Sub:TMCP cond} Lastly, following  \cite{cavalletti2017} for metric measure spaces,  we establish a version of \autoref{Th:Linfty estimates} for the subsequent \emph{timelike measure contraction property} $\smash{\TMCP_p^e(K,\infty)}$  \cite[Def.~3.7]{cavalletti2020} as follows.

\begin{definition}\label{Def:TMCP} A measured Lorentzian pre-length space $(\mms,\met,\meas,\ll,\leq,\tau)$ satisfies $\smash{\TMCP_p^e(K,N)}$ for  $p\in (0,1)$, $K\in \R$ and $N\in (0,\infty)$ if for every $\mu_0\in\scrP_\comp(\mms)\cap\Dom(\Ent_\meas)$ and every $x_1\in I^+(\mu_0)$ there exists a timelike proper-time parametrized $\smash{\ell_p}$-geodesic $(\mu_t)_{t\in [0,1]}$ from $\mu_0$ to $\mu_1 := \delta_{x_1}$ such that for every $t\in[0,1)$,
\begin{align*}
\scrU_N(\mu_t) \geq \sigma_{K,N}^{(1-t)}(T_{\mu_0\otimes\mu_1})\,\scrU_N(\mu_0).
\end{align*}
\end{definition}

\begin{lemma}\label{Le:Hurr} Let $(\mms,\met,\meas,\ll,\leq,\tau)$ satisfy $\smash{\TMCP_p^e(K,N)}$ for $p\in (0,1)$, $K\in\R$, and $N\in (0,\infty)$. Suppose that $\mu_0=\rho_0\,\meas\in\scrP_\comp^\ac(\mms,\meas)$ with $\rho_0\in\Ell^\infty(\mms,\meas)$. Lastly, let $x_1 \in I^+(\mu_0)$. Then there exists some timelike proper-time parametrized $\smash{\ell_p}$-geodesic $(\mu_t)_{t\in[0,1]}$ from $\mu_0$ to $\smash{\mu_1:=\delta_{x_1}}$ such that for every $t\in (0,1)$, we have $\mu_t = \rho_t\,\meas\in\Dom(\Ent_\meas)$ and 
\begin{align*}
\meas\big[\{\rho_t > 0\}\big]\geq (1-t)^N\,\rme^{-tD\sqrt{K^-N}}\,\big\Vert\rho_0\big\Vert_{\Ell^\infty(\mms,\meas)}^{-1},
\end{align*}
where $\smash{D\geq \sup\tau(\supp\mu_0\times\{x_1\})}$.\end{lemma}

\begin{proof} We may and will confine ourselves to the case $K\leq 0$. Given a timelike proper-time parametrized $\smash{\ell_p}$-geodesic $(\mu_t)_{t\in[0,1]}$ from $\mu_0$ to $\mu_1$ obeying the  inequality defining $\smash{\TMCP_p^e(K,N)}$ and arguing similarly as in the proof of \autoref{Le:Bounded midpoints},
\begin{align*}
\meas\big[\{\rho_t >0\}\big]^{1/N} &\geq \scrU_N(\mu_t) \geq \sigma_{K,N}^{(1-t)}(T_{\mu_0\otimes\mu_1})\,\scrU_N(\mu_0)\\ &\geq (1-t)\,\rme^{-tD\sqrt{-K/N}}\,\big\Vert\rho_0\big\Vert_{\Ell^\infty(\mms,\meas)}^{-1/N}.
\end{align*}
This  terminates the proof.
\end{proof}

\begin{theorem}\label{Th:TMCP} Assume $\smash{\TMCP_p^e(K,N)}$ for some $p\in (0,1)$, $K\in\R$, and $N\in (0,\infty)$. Let $\mu_0=\rho_0\,\meas\in\scrP_\comp(\mms)$ with $\rho_0\in\Ell^\infty(\mms,\meas)$.  Lastly, let $x_1\in I^+(\mu_0)$. Then there exists a timelike proper-time parametrized $\smash{\ell_p}$-geodesic $(\mu_t)_{t\in[0,1]}$ from $\mu_0$ to $\smash{\mu_1 := \delta_{x_1}}$ satisfying the following two properties for every $t\in [0,1)$.
\begin{enumerate}[label=\textnormal{(\roman*)}]
\item We have $\mu_t=\rho_t\,\meas\in\Dom(\Ent_\meas)$ with
\begin{align}\label{Eq:Sem}
\scrU_N(\mu_t) \geq \sigma_{K,N}^{(1-t)}(T_{\mu_0\otimes \mu_1})\,\scrU_N(\mu_0). 
\end{align}
\item Setting $D:= \sup\tau(\supp\mu_0\times\{x_1\})$, we have
\begin{align*}
\Vert\rho_t\Vert_{\Ell^\infty(\mms,\meas)} \leq \frac{1}{(1-t)^N}\,\rme^{Dt\sqrt{K^-N}}\,\Vert\rho_0\Vert_{\Ell^\infty(\mms,\meas)}.
\end{align*}
\end{enumerate}
\end{theorem}

\begin{proof} The bisection argument from \autoref{Ch:Good} does not work  under  $\smash{\TMCP_p^e(K,N)}$ since $\mu_1\not\ll\meas$, while every intermediate measure should be absolutely continuous with respect to $\meas$. We rather follow the proof of  \cite[Thm.~3.1]{cavalletti2017}. 

Let $\smash{(\mms, \met, \ll^{\leftarrow}, \leq^{\leftarrow}, \tau^{\leftarrow})}$ denote the causally reversed structure of $(\mms,\met,\ll,\leq,\tau)$ \cite[Def.~1.2]{cavalletti2020}, i.e.~$\smash{x\ll^\leftarrow y}$ if and only if $y\ll x$, $\smash{x\leq^\leftarrow y}$ if and only if $y\leq x$, and $\smash{\tau^\leftarrow(x,y) := \tau(y,x)}$, $x,y\in\mms$. Let $\smash{\ell_p^\leftarrow}$ be the cost function associated to $\smash{\tau^\leftarrow}$.

\textbf{Step 1.} \textit{Construction of a ``backward'' geodesic.} Given any $n,k\in\N$, we set $\smash{s_n^k := (1-2^{-n})^k}$. Fix $n\in\N$, and assume that $\smash{\bdbeta_n^k\in \OptTGeo_{\ell_p^\leftarrow}^{\tau^\leftarrow}(\mu_1,\mu_0)}$ has been defined such that for every $i\in\{1,\dots,k\}$, $\smash{(\eval_{s_n^i})_\push\bdbeta_n^k \in\scrP_\comp^\ac(\mms,\meas)}$ and
\begin{align*}
\sup\tau^\leftarrow(\{x_1\}\times \supp\,(\eval_{s_n^i})_\push\bdbeta_n^k)\leq 2^{-n}\,s_n^{i-1}\,D.
\end{align*}
Let the functional $\smash{\scrV_N^{2^{-n}}}$ be defined as in \eqref{Eq:Functional E}. By \autoref{Le:Maximizer}, the latter admits a maximizer $\smash{\bdpi_n^{k+1}\in\OptTGeo_{\ell_p}^\tau((\eval_{s_n^k})_\push\bdbeta_n^k,\mu_1)}$. Let $\bdsigma_{k+1}^n\in\smash{\OptTGeo_{\ell_p^\leftarrow}^{\tau^\leftarrow}(\mu_1,(\eval_{s_n^k})_\push\bdbeta_n^k)}$ be the timelike $\smash{\ell_p}$-optimal geodesic plan obtained by ``time-reversal'' of $\smash{\bdpi_n^{k+1}}$. By a gluing argument, we construct a measure $\smash{\bdbeta_n^{k+1}\in\OptTGeo_{\ell_p^\leftarrow}^{\tau^\leftarrow}(\mu_1,\mu_0)}$ with
\begin{align*}
(\Restr_0^{s_n^k})_\push\bdbeta_n^{k+1} &= \bdpi_n^{k+1},\\
(\Restr_{s_n^k}^1)_\push\bdbeta_n^{k+1} &= (\Restr_{s_n^k}^1)_\push\bdbeta_n^k.
\end{align*}
Using the induction hypothesis, the bound
\begin{align*}
\sup \tau^\leftarrow(\{x_1\}\times\supp\,(\eval_{s_n^{k+1}})_\push\bdbeta_n^{k+1})\leq 2^{-n}\,s_n^k\,D
\end{align*}
obtained by construction, using \autoref{Le:Hurr} and  following the lines of \autoref{Pr:Min 0}, \autoref{Cor:Thr} and \autoref{Pr:Max = 0} for the threshold
\begin{align*}
c_n^{k+1} := \frac{1}{(1-2^{-n})^N}\,\rme^{2^{-n}s_n^kD\sqrt{K^-N}}\,\Vert\rho_0\Vert_{\Ell^\infty(\mms,\meas)}
\end{align*} 
for every $n,k\in\N$ we obtain $\smash{(\eval_{s_n^{k+1}})_\push\bdbeta_n^{k+1} = \rho_{s_n^{k+1}}\,\meas\in\Dom(\Ent_\meas)\cap\scrP_\comp(\mms)}$ with
\begin{align*}
\big\Vert\rho_{s_n^k}\big\Vert_{\Ell^\infty(\mms,\meas)} \leq \frac{1}{(1-2^{-n})^N}\,\rme^{2^{-n}s_n^{k-1}D\sqrt{K^-N}}\,\big\Vert \rho_{s_n^{k-1}}\big\Vert_{\Ell^\infty(\mms,\meas)},
\end{align*}
Inductively, by geometric summation this yields
\begin{align}\label{Eq:ZUB}
\big\Vert\rho_{s_n^k}\big\Vert_{\Ell^\infty(\mms,\meas)} \leq \frac{1}{(s_n^k)^N}\,\rme^{(1-s_n^k)D\sqrt{K^-N}}\,\Vert\rho_0\Vert_{\Ell^\infty(\mms,\meas)}.
\end{align}

\textbf{Step 2.} \textit{Construction of the  geodesic and verification of its properties.} We iteratively construct a family 
\begin{align*}
\{\bdbeta_n^k : n \in\N,\, k\in\N_0\}\subset\OptTGeo_{\ell_p^\leftarrow}^{\tau^\leftarrow}(\mu_1,\mu_0)
\end{align*}
according to Step 1. Let $(\bdbeta^n)_{n\in\N}$ be an enumeration of the elements of this class. By \autoref{Le:Plans}, the latter admits a weak limit $\smash{\bdbeta\in\OptTGeo_{\ell_p^\leftarrow}^{\tau^\leftarrow}(\mu_1,\mu_0)}$ along a nonrelabeled subsequence. Let $\smash{\bdalpha\in\OptTGeo_{\ell_p}^\tau(\mu_0,\mu_1)}$ be the ``time-reversal'' of $\bdbeta$, which induces a timelike proper-time parametrized $\smash{\ell_p}$-geodesic $(\mu_t)_{t\in[0,1]}$ from $\mu_0$ to $\mu_1$. By weak lower semicontinuity of $\scrF_c$ in $\scrP(J(\mu_0,\mu_1))$ for appropriate values $c>0$, we get $\mu_t = \rho_t\,\meas\in\Dom(\Ent_\meas)\cap\scrP_\comp(\mms)$, and as in the last step of the proof of \autoref{Th:Ent linfty} we obtain the weak stability of \eqref{Eq:ZUB}, whence $\Vert\rho_t\Vert_{\Ell^\infty(\mms,\meas)}$ obeys the desired estimate for every $t\in[0,1)$. The proof is terminated.
\end{proof}

\begin{remark} If in addition $(\mms,\met,\meas,\ll,\leq,\tau)$ is timelike non\-branching in the above  \autoref{Th:TMCP}, as for \autoref{Cor:Timelike branching} $\mu_0$ and $\mu_1$ are connected by a unique timelike proper-time parametrized $\smash{\ell_p}$-geodesic, which thus automatically satisfies the conclusions of \autoref{Th:TMCP}.
\end{remark}

\end{document}